\documentclass[format=acmsmall, review=false]{acmart}
\usepackage{acm-ec-23}
\usepackage{booktabs} 
\usepackage[ruled]{algorithm2e} 

\SetAlFnt{\small}
\SetAlCapFnt{\small}
\SetAlCapNameFnt{\small}
\SetAlCapHSkip{0pt}
\IncMargin{-\parindent}

\setcitestyle{authoryear}

\newtheorem{remark}{Remark}
\newtheorem{axiom}{Axiom}

\title[Grading and ranking a large number of candidates]{Grading and ranking a large number of candidates}

\author{Rida Laraki and Estelle Varloot}

\begin{abstract}

It is common that a jury must grade a set of candidates in a cardinal scale such as $\{1,2,3,4,5\}$ or an ordinal scale such as $\{$Great, Good, Average, Bad$\}$. When the number of candidates is very large such as hotels (BOOKING), restaurants (GOOGLE), apartments (AIRBNB), drivers (UBER), or papers (EC), it is unreasonable to assume that each jury member will provide a separate grade for each candidate. Each jury member is more likely to abstain for some candidates, cast a blank vote, or be associated at random, or as a function of its expertise, with only a small subset of the candidates and is asked to grade each of those. 

Extending the classical theory, we study aggregation methods in which a voter will not be eligible to grade all the candidates, and the candidates are not eligible for the same sets of voters. Moreover, each candidate on which they are eligible, the voter will have the choice between: a blank vote, grade the candidate, or abstain. 

Assuming single-peaked preferences over the grades, we axiomatically characterise a broad class of strategy-proof grading mechanisms satisfying axioms such as unanimity, anonymity, neutrality, participation or consistency. Finally, when a strict ranking is necessary (to distinguish let say between two borderline papers in a conference), some tie-breaking rules, extending the leximin and majority judgment, are defined and are shown to be equivalent to some strategy-proof grading functions on a richer space of outcome.

Our paper will propose new rules, called phantom-proxy mechanisms, to aggregate the votes in the examples above or others, which differ from the usual average mark, that are easily manipulable. Moreover, the phantom-proxy are able to reduce the injustices caused by some candidates juries too generous or severe.

\end{abstract}

\begin{document}

\begin{titlepage}

\maketitle

\end{titlepage}

\section{Introduction}

Have you ever waken up a week before a political election day to find there are 78 candidates on the ballot but you don't know much about most of them? You may have an opinion about five of them, and think two are good and three are bad, but you just don't have the time and the relevant information to study them all. If the ballot asks you to strictly rank all the 78 candidates, it is probably much too complicated. If the ballot asks you to rank just the one you know, this is not a good idea unless a clear meaning is given that the two top are good and the remaining three are bad. If the ballot specifies that you only rank the candidates you think are good then, you should only rank two of the five, but then you are denied giving your opinion about the three you think are bad!

The reader might find this unlikely, but this was precisely the situation voters in Australia faced for the 2004 Senate elections in the state of New South Wales. The ballot asks the voters to rank strictly all the candidates, otherwise, the vote is considered invalid. To solve the issue, voters were allowed to “vote above the line” or “vote below the line.” where “Above” means they choose the ranking specified by one party and “Below” means the voter determines their own ranking. This leads most Australians (95\%) to vote above the line, which implies that the outcome is decided by a strategic game played by the parties. We propose a family of alternative solutions where each voter is asked to grade as many candidates as they want to. A method in this family has been adopted in 2020 by Paris for its participatory budget,\footnote{Voters gave a grade to as many projects as they want in a scale of four grades: « Coup de coeur/J’adore, J'aime bien / C'est intéressant,  Pourquoi pas,  Je ne suis pas convaincu ». In 2022, 82 million euros was allocated to 62 projects. About 150.000 voters participated (about 130.000 with paper ballots and 20.000 electronically). All projects received at least 1000 grades (some much more), which is statistically representative, and more informative compared to the previous system where several projects received less than a hundred votes.} where the sum of the grades received by a project are normalized to $100\%$ then majority judgment (MJ) ranks them \cite{BL2011}. 


A participatory democratic initiative was proposed by \href{https://laprimaire.org/election-presidentielle-2017/}{LaPrimaire.org}.\footnote{https://laprimaire.org/election-presidentielle-2017/} About 150.000 French voters participated in a process\footnote{https://laprimaire.org/deroulement/} to nominate a citizen as a candidate for the 2017 French presidential election. An initial slate of about 200 candidates was whittled down to the 12 who were supported by at least 500 voters. Then each voter was asked to evaluate (with a few days delay) five candidates out of the 12 on the scale Excellent, Very Good, Good, Passable, Insufficient. The assignment of five was done randomly\footnote{One of the objectives of the random selection process was to guarantee that all the candidates are evaluated by approximately the same number of voters and that this number is large enough for the results to be statistically representative, which was the case: each candidate received in average 4454 voters, with a minimum on 4372 and a maximum of 4513. See https://articles.laprimaire.org/résultats-du-1er-tour-de-laprimaire-org-c8fe612b64cb } and the twelve ranked by MJ after the $100\%$ normalization. The reason for not asking to vote on all the 12 was to incite them to invest time to make a careful comparative study.\footnote{Each candidate wrote a political program with several documents and videos. See for instance the program of the winning candidate Charlotte Marchandise:  https://laprimaire.org/qualifie/charlotte-marchandise-franquet}  

Motivated by the above applications, our paper studies mechanisms where voters have different, exogenously, rights to grade subsets of the set of candidates. The optimal process to allocate the rights is not studied in this paper and is delegated to a future work. In practice and depending on the application (large or small electorate), it can be done at random as in LaPrimaire.org, as a function of the expertise or of conflict of interest, or as a combination of the previous as in CS conferences. Our mechanisms will allow the voters to abstain or vote blank, as in Paris participatory budgeting.


Following \cite{M1980} and \cite{BL02007,BL2011,Balinski-Laraki} framework, our paper study the grading methods that associates a final grade to each candidate, given the choices of all the voters on their eligible candidates, where the choices a voter has on a candidate are: a blank vote, a grade, or to abstain. We also propose tie-breaking-rules, whenever a strict ranking is needed (borderline papers as in a CS conference).
Assuming single-peaked preferences over the final grades, we characterize a class of strategy-proof grading functions satisfying familiar axioms such as unanimity, anonymity, neutrality, participation, or (variable electorate) consistency,  see \cite{Arrow,BCELD,M988}. But, contrarily to classical theory, there are many ways to define some axioms such as neutrality or anonymity because voters don’t have the same rights and candidates are not eligible for the same voters. 

\textbf{Mains contributions and literature}. Technically, our paper uses and is an extension of \cite{M1980} where the inputs from voters are interpreted as elements of an ordered set of grades. \cite{M1980} interpreted in our model, characterized anonymous and non-anonymous strategy-proof methods when $n$ voters grade one candidate. Our innovations compared to his work are:
\begin{itemize}
    \item We deal with $n$ voters and $m$ candidates where a candidate is eligible for a subset of voters, and a voter may cast a blank vote or abstain for each candidate.
    \item We assume the set of input grades to be smaller than the set of outputs.
    \item We characterize all strategy-proof grading functions combined with other axioms not considered in \cite{M1980} and illustrate our results on the phantom-proxy class.
    \item We propose several formulations of anonymity and show that the definition has an impact on the characterizations. A similar conclusion holds for other axioms.
    \item We extend the grading methods to rank, by associating them to tie-breaking rules, in the spirit of majority judgment and show that the ranking function may be interpreted as a strategy-proof grading function with richer output space. 
    \end{itemize}

Moulin's paper inspired a large literature that obtained characterizations for other domains or proved impossibility results (see, among many others,\cite{BJ83}, \cite{NP2007}, \cite{Barbera} \cite{ICML2016} and \cite{EC2019}).  The later article for example introduces the class of phantom moving mechanisms which are anonymous, neutral and strategy-proof in the budget aggregation problem. Extensions that cover the case of private consumption or capacity constraints have been studied
by \cite{Moulin2017} and \cite{Aziz2019}. \cite{LV2022} extends Moulin's phantom mechanisms to expert aggregation problems where each voter submits a probability distribution (a prediction) over an ordered set.

Operationally and conceptually, our paper extends the majority judgment theory of \cite{BL02007,BL2011}. They introduced and studied a class of grading and ranking methods that avoids Arrow and Condorcet paradoxes, and are resistant to strategic manipulations. In their model, $n$ voters (equally treated) are requested to grade $m$ candidates (equally treated), and the output is a final grade for each candidate and a ranking of the candidate. Our paper extends their main results and methods to situations where voters don't have necessarily the same weights or rights to vote, candidates and voters are not necessarily treated equally, and voters can abstain or cast blank votes.

We also contribute to the literature on incomplete preferences in voting. see \cite{BRM2016} for a survey. For example, \cite{BS2020} approximate Borda and minmax rules in a context with a large electorate where each voter is asked to rank a random subset of $l\geq 2$ candidates, or to provide a ranking of her l most preferred candidates. 
\cite{KL2005} introduced the notions of possible and necessary winners (PW, and NW) for a voting function $f$ when we have access to partial information (a partial ranking for example) and computed the complexity of determining whether a candidate is PW or NW. \cite{BB2011} studied the minimax regret instead of PW or NW. Our problem is quite different: we search for a good rule to aggregate the grades when our partial information is caused by some candidates being ineligible for some voters, and some voters not having any opinion on the candidates they are eligible for and prefer to abstain or vote blank or they have an opinion but prefer to abstain strategically. 

\textbf{Structure of the paper}: Section 2 contains all the notations needed to follow the paper. Section 3 studies the incentive-compatible methods when a voter think of manipulating by changing his input grade, introduces the class of phantom-proxy mechanisms, and illustrates the general characterization of the proxy class. Section 4 studies the incentives for a voter implied by having, for each candidate, options other than grading: blank votes and abstention. We define and link the following properties: how to count Blank Votes (BV), Silent Ignored (SI), Silent Consent (SC), two forms of Participation (P and FP), Jury Determinism (JD), and Strong strategy-proofness. Section 5 characterizes the SP methods that satisfy additional properties such as Unanimity (U), three forms of Neutrality (N, SN and F), two of Anonymity (A and SA), and two of variable electorate Consistency (OC and IC). Section 6 extends each (OC,F) proxy grading function to a ranking function with almost no ties and shows that it can be interpreted as a strategy-proof grading function in a richer output space. Section 7 discusses some extensions and concludes. An appendix contains the missing proofs.

\section{Notations}
There are $3$ main outputs we will be interested in. The final (or aggregate) grade a candidate gets, the tie breaking rule to rank any two candidates that get the same grade, and the final ranking among all the candidates. As voters will have different rights, can abstain or vote blank, we will have to deal with too many subsets. To make it easy for the reader, we put all the notations here.
\begin{itemize}
    \item Voters will be described with small cap letters (e.g $i,j$) and the set of voters is $\mathcal{N}$. The set is assumed finite.
    \item Candidates will be described with capital letters (e.g $I,J$) and the set of candidates is $\mathcal{M}$. 
    \item For a voter $i$, let $\mathcal{C}_i \subseteq \mathcal{M}$ be the finite set of candidates on which voter $i$ is allowed to vote.
    \item $\mathcal{C'}_i \subseteq \mathcal{C}_i$ is the set of candidates voter $i$ provided a grade for. 
    \item For a candidate $J$, let $\mathcal{D}_J \subseteq \mathcal{N}$ be the set of voters that were asked to grade $J$. 
    \item $\mathcal{D''}_J(\textbf{v})$ is the set of voters that gave a grade to candidate $J$ in the voting profile $\textbf{v}$. When the context is clear we will use the shorthand $\mathcal{D''}_J$.
    \item Grades will use Greek letters (e.g $\alpha,\beta$). The set of grades that can be expressed by voters ( inputs) is $\mathcal{A}$. The set of grades that can be provided by the grading function (outputs) is $\mathcal{B}$. As it is the case in practice, we assume $\mathcal{A} \subseteq \mathcal{B}$ and that a total order exists on $\mathcal{B}$. 
    An example is $\mathcal{A}=\{1,2,3,4,5\}$ and $\mathcal{B}=[1,5]$. Without loss of generality we consider that $\{\inf\mathcal{A},\sup\mathcal{A}\} \subseteq \mathcal{B}$.
    \item The \textbf{ballot} of voter $i$ is $v_i$. 
    The notation $v_i(J)$ corresponds to the opinion voter $i$ expressed about candidate $J$. It is the \textbf{vote} of player $i$ for $J$. We also use the shorthand $\textbf{v}(J)$ to describe the multi-set (bag) of \textbf{grades} provided for $J$.
    \item We also need notations to describe situations that do not correspond to a voter supplying a grade to a candidate. As such when $i \not \in \mathcal{C}_i$ we denote his vote as $v_i(J)=\emptyset$ (the absence of a possible expressed opinion). For voters who to submit a blank vote for a given candidate we use the notation $v_i(J)=\otimes$. We will distinguish these from voters who chose to ignore the opportunity to submit a vote, these are referred to as absentees and will be denoted with $v_i(J)=\circ$. We denote $\mathcal{E} = \{\emptyset,\circ,\otimes\}$ all the possible inputs that are not grades. We denote $\mathcal{O} = \mathcal{A} \cup \mathcal{E}$ all the possible submitted votes and the absence of possible vote. 
    \item The voting profile containing all the votes for the election is $\textbf{v} \in \mathcal{O}^{\mathcal{M}\times\mathcal{N}}$. We describe as $\textbf{v}_{-i}$ the voting profile where voter $i$ was removed.
\item The notation $\textbf{v}_{-T}$ refers to the profile obtained from $\textbf{v}$ where all voters $i$ that are in $T\subseteq M$ have their allowed votes replaced by $\otimes$.
    \item For any $\alpha \in \mathcal{O}$. We will also use the notation $v_i[J:=\alpha]$ that denotes the ballot of $i$ where the vote that $i$ provided for candidate $J$ was replaced by $\alpha$ everything else equal. Similarly $\textbf{v}[v_i:=w_i]$ represent the situation voting profile where the ballot of $i$ was replaced by $w_i$ and $\textbf{v}[v_i(J):=\alpha]$ the voting profile where the vote of $i$ was replaced by $\alpha$ everything else equal.
    \item We are therefore looking for an \textbf{grading function} $\varphi : (\mathcal{O})^{\mathcal{M}\times\mathcal{N}} \rightarrow \mathcal{B}^\mathcal{M}$.


\end{itemize}



\section{Strategy-proofness in grading : SP}\label{section SP}
\subsection{The general minmax characterization with the phantom mappings}


We will assume that each voter's objective is to try to make the outcome for any candidate (determined by the inputs and aggregation rule) be as close as possible to her (true) grade for that candidate. Hence, strategy-proof --in grading-- can be described as follows.

\begin{definition}[Strategy-proof in Grading : SP]
A grading function $\varphi: \mathcal{O}^{\mathcal{M} \times \mathcal{N}} \rightarrow (\mathcal{B} \cup \{\emptyset\})^\mathcal{M}$ is strategy-proof (or is incentive compatible) in grading  if for any voter $i$, for any candidate $J \in \mathcal{C}_i$, and for any $v_i$ and $w_i$ such that $v_i(J) \in \mathcal{A}$ and $w_i(J) \in \mathcal{A}$ we have:
\[\varphi(\textbf{v})(J) > v_i(J) \Rightarrow  \varphi(\textbf{v}[v_i:=w_i])(J) \geq \varphi(\textbf{v})(J)\]
\[\varphi(\textbf{v})(J) < v_i(J) \Rightarrow  \varphi(\textbf{v}[v_i:=w_i])(J) \leq \varphi(\textbf{v})(J)\]

\end{definition}

In other words, a voter who graded a candidate shouldn't be able to change the outcome for that candidate closer to his input grade by lying. It is interesting that this definition implies that a voter does not question its impact on the candidates they are not allowed to vote for. We will come back to this question later when we will study stronger notions of SP.


Also, strategy-proofness as defined above does not prevent a vote by $i$ for a candidate $I$ from impacting the grade for on another candidate $J$ if $i$ is not eligible for $J$. However, $i$ cannot impact $J$ if allowed to vote for as the following useful lemma shows.

\begin{lemma}\label{SP lemma}
SP implies that if $i$ graded $I$ and $J$, its vote about $I$ does not impact the outcome for $J$. 
\end{lemma}

\begin{proof}
Let $i$ be allowed to vote for both $I$ and $J$. Let $w_i(J) = v_i(J)$ and $w_i(I) \neq v_i(I)$ and everything else equal.

Suppose that $\varphi(\textbf{v})(J) > v_i(J)$. 
\[\varphi(\textbf{v})(J) > v_i(J) \Rightarrow  \varphi(\textbf{v}[v_i:=w_i])(J) \geq \varphi(\textbf{v})(J).\]
Therefore $\varphi(\textbf{v}[v_i:=w_i])(J) > w_i(J)$. It follows that $\varphi(\textbf{v})(J) \geq \varphi(\textbf{v}[v_i:=w_i])(J) \geq \varphi(\textbf{v})(J)$, therefore:
\[\varphi(\textbf{v}[v_i:=w_i])(J) =\varphi(\textbf{v})(J).\]
Conversely for $\varphi(\textbf{w})(J) > w_i(J)$, $\varphi(\textbf{v})(J) < v_i(J)$ and $\varphi(\textbf{w})(J) < w_i(J)$. As such $\varphi(\textbf{v})(J) = \varphi(\textbf{v}[v_i:=w_i])(J)$. 
\end{proof}

We can characterize now all the SP grading functions. The proof is a direct consequence of \cite{M1980} and the above lemma.

\begin{theorem}[SP general characterization]\label{SP theo}
If a grading function $\varphi: \mathcal{O}^{\mathcal{M} \times \mathcal{N}} \rightarrow (\mathcal{B} \cup \{\emptyset\})^\mathcal{M}$ is strategy-proof (SP) then for any $J$ there are $ 2^{\#\mathcal{D}_J}$ functions $\omega_{J,S}^{T} : \mathcal{O}^{\mathcal{M} \times \mathcal{N}} \rightarrow \mathcal{B}$ such that for all $\emptyset \subseteq S \subseteq S'\subseteq T \subseteq \mathcal{D}_J$ we have  $\omega_{J,S}^{T} \leq \omega_{J,S'}^{T}$ and:
\[\forall \textbf{v}, \varphi(\textbf{v})(J)=
\max_{\emptyset \subseteq S \subseteq \mathcal{D''}_J}
    \min (\{v_i(J): i \in S\} \cup\{\omega_{J,S}^{\mathcal{D''}_J}
    (\textbf{v}_{-\mathcal{D''}_J})\})\]
and we will call the $\omega_{J,S}^T$ the "phantom-mappings".
\end{theorem}

\begin{proof}
Let us fix $J$. Let us fix the ballots of all voters that do not provide a grade for $J$. Then by the Moulin's theorem (proposition 3 in the paper \cite{M1980}), 
there exist $2^n$ constants $\alpha_S$ (called phantoms) such that:
\[\varphi(\textbf{v})(J)=\max_{\emptyset \subseteq S \subseteq \mathcal{D''}_J}\min (\{\alpha_S\} \cup \{v_i(J): i \in S\}).\]

It follows that when we no longer consider that the ballots of voter that did not provide a grade for $J$ are fixed, we replace the $\alpha_S$ by functions $\omega_{J,S}^T$. According to lemma \ref{SP lemma}, the $\omega_{J,S}^T$ functions only depend on the ballots of voters that did not grade $J$. As clearly any $\varphi$ of this form is (SP), we have characterised the group of methods.

\end{proof}

The proof looks easy to prove, but only because we are using the non-trivial result of \cite{M1980}. We will refine and simplify in the sequel the maxmin formula by adding axioms until we obtain something similar to the famous median formula of \cite{M1980} for strategy-proof rules in the anonymous case. But before that, we introduce the SP family of phantom-proxy mechanisms, that can be described using the order statistics, and which will be our leading example.

When $\mathcal{A}=\{0,1,2,3,4,5\}$ and $\mathcal{B}=[0,5]$, the restriction in the input space is likely due to a desire to keep the process simple for voters even if in theory the regulator wouldn't have minded giving the voters more freedom. It makes sense to consider that the regulator chose a mechanism in $\Psi :\mathcal{B} \rightarrow \mathcal{B}$ based on its good properties and then restricted the possible input votes to $\mathcal{A}$. It is interesting to note that (SP) mechanisms can be extended to a (SP) mechanism $\Psi :\{\mathcal{B} \cup \mathcal{E}\}^{\mathcal{M} \times \mathcal{N}} \rightarrow (\mathcal{B} \cup \{\emptyset\})^\mathcal{M}$. Any arbitrary extension of the $\omega_{J,S}^T$ phantom mappings will do.  However, we provide most of the characterizations for $\mathcal{A} \subset \mathcal{B}$, unless the characterization is more elegant with $\mathcal{A} = \mathcal{B}$.


\subsection{The phantom-proxy grading functions}
Let us now consider a situation where a group of voters (the jury) is expected to grade a large number of candidates. For example, the grading of the GCSE. 
The regulator may learn for that jury votes what are their criteria regarding the candidates they grade and therefore use that be predict which grade they would give to all the candidates. This can be used to reduce the bias caused by the diversities of the juries and that some of them are more generous than others. In this subsection, we introduce a class of methods that are based on this idea: \textbf{The phantom-proxy mechanisms}.

The concept for these mechanisms are simple. For a given candidate $J$, if a voter $i$ submitted a blank vote or was not allowed to vote ($v_i(J) \in \otimes,\emptyset$) then a proxy-vote based on what we know about the voter and the candidate will be designated to replace his vote. Absentee votes are either provided another proxy or removed from the process entirely. We therefore have for each candidate $J$ a multi-set of grades that each correspond to a different voter, either because they graded $J$ or because we have a proxy vote that represents them. We call this multi-set the \textbf{voting pool} for $J$. Once we have this pool, we can simply select one of its elements in a non-bias way. To do this we use the order functions.

\begin{definition}[Order functions]
The order function $\mu_k : \mathcal{B}^\mathcal{N} \rightarrow \mathcal{B}$ is the function that takes a multi-set (also known as a bag) as its input and returns the $k$-th smallest element of that multi-set.
\end{definition}

\begin{remark}\label{OF remark}
The order functions are the only single-peaked aggregation functions that always output one of their inputs and whose output doesn't change when the inputs are permuted, (\cite{BL2011}, chapter 11)  but this is obvious from \cite{M1980} median anonymous characterization).
\end{remark}

This characterization tell us that order functions are the only one that guarantee to select a vote in the available voting pool without any bias in regards to which voter is associated to which vote.

\begin{definition}[Phantom-proxy mechanisms]
A mechanism $\psi: \mathcal{O}^{\mathcal{M} \times \mathcal{N}} \rightarrow (\mathcal{B} \cup \{\emptyset\})^\mathcal{M}$ is phantom-proxy if for each voter $i$ and each candidate $J$ there is a function $f_{i,J}: \mathcal{O}^\mathcal{N} \rightarrow  (\mathcal{B} \cup \{\emptyset\})$ s.t.
\[v_i(J)\in \mathcal{A}\wedge (\forall J, v_i(J) \in \{\emptyset,\otimes\} \Rightarrow f_{i,J}(v_i)) =\emptyset\]
and a function $g_J: \mathbf{N} \rightarrow \mathbf{N}$ where for all $k$, $0 < g_J(k) \leq k$ such that:
\[\varphi(\textbf{v})(J)=\mu_{g_J(\#\mathcal{D''}_J \cup \mathcal{F}_J(\textbf{v}))}(\textbf{v}(J) \cup \mathcal{F}_J(\textbf{v}))\] 
Where $\mathcal{F}_J(\textbf{v}) = \{f_{i,J}(v_i)| f_{i,J}(v_i) \neq \emptyset\}$ is the multi-set containing the proxy votes for $J$.
\end{definition}

We use the notation $\mathcal{F}_J$ to represents the function that takes a voting profile $\textbf{v}$ and returns $\mathcal{F}_J(\textbf{v})$. 

With an ordinal scale, a well-known phantom-proxy grading function is the so-called majority grade in \cite{BL2011,BL02007}. This is a grading method where each candidate's grade is the smallest median value of the votes it received. It corresponds to the phantom-proxy mechanism where all the $f_{i,J}$ functions are constants equal to $\emptyset$ and the $g_J$ function is $g_J(k)=\lfloor k/2\rfloor$.

With a cardinal scale, an intuitive phantom-proxy grading function $\psi: \mathcal{O}^{\mathcal{M} \times \mathcal{N}} \rightarrow (\mathcal{B} \cup \{\emptyset\})^\mathcal{M}$ would be for all voters to be replaced by the average of the grades they gave and not to have a proxy if they never gave a grade: 
\[\forall v_i, \exists I, v_i(I)\in \alpha \Rightarrow f_{i,J}(v_i) = \dfrac{1}{\#\mathcal{C'}_i}\sum_{J \in \mathcal{D''}_J}v_i(J).\]
\[\forall v_i, \forall I, v_i(I)\in \mathcal{E} \Rightarrow f_{i,J}(v_i) = \emptyset.\]

For example, with $\mathcal{N}=\{x,y,z\}$ and $\mathcal{M}=\{I,J\}$ then for $v_x =(1,\emptyset)$, $v_y =(\emptyset,3)$, $v_y=(2,2)$ with $g_I=\min$ order function and $g_J=\max$ order function, we obtain $\varphi(\textbf{v})=(1,3)$.




\begin{proposition}\label{SP prop}
All phantom proxy-mechanisms $\psi$ are SP. 
\end{proposition}

The proof of the proposition is in the appendix (\ref{SP prop proof}). Since all phantom-proxy mechanisms are SP it follows that we can characterize them using the phantom-mapping functions. 

\begin{proposition}[Characterization of phantom-proxy]\label{SP charac prop}
For any $S$ and $T$ and $\textbf{v}$, let $p=g_J(\#T+ \#\mathcal{F}_J(\textbf{v}))$. Let $k= \#S - \#T +p$

\begin{itemize}
    \item If $k \leq 0$ then $\omega_{J,S}^T(\textbf{v}_{-T}) = \inf\mathcal{B}$.
    \item If $k > \#\mathcal{F}_J(\textbf{v})$ then $ \omega_{J,S}^T(\textbf{v}_{-T})  = \sup\mathcal{B}$
    \item Else $\omega_{J,S}^T(\textbf{v}_{-T}) = \mu_{k}(\mathcal{F}_J(\textbf{v}))$.
\end{itemize}
(Proof \ref{SP charc prop proof}.)
\end{proposition}

Intuitively, since the only values that can be selected aside from the inputs in a grading function are the phantom-mapping and in a phantom-proxy mechanism are the proxy votes we have that all phantom-mappings outcomes must be associated to one of the proxy votes or never get selected by the mechanism. Take an example: if $g_J$ is such that we must select the smallest element in the voting pool then if all grades in $\textbf{v}(J)$ are larger or equal to the smallest proxy vote then that proxy vote is selected. Therefore, we must have $\omega_{J,T}^T(\textbf{v}_{-T})=\min\mathcal{F}_J(\textbf{v})$. Otherwise, if there is at least one grade from $\textbf{v}(J)$ that is less than the smallest proxy vote then no proxy votes can be selected as such $\omega_{J,S}^T(\textbf{v}_{-T}) < \inf\mathcal{A}$. Since we need to be certain that $\omega_{J,S}^T \leq \omega_{J,T}^T$ we can select $\omega_{J,S}^T(\textbf{v}_{-T}) =\inf\mathcal{B}$.

\section{Incentives and Impact of Non-Graders}\label{section non graders}

There are numerous ways we can treat blank votes and absentees. This paper will make suggestions we believe make sense when they are to be treated differently. It is important to note that the regulator may choose 
to count them similarly as it is often done in practice where blank votes and absentees are ignored and so, have the same (no) impact on the outcome of the election. 


In political elections, blank vote symbolically represents protest votes. In a grading system however if you want to protest against a candidate, you can just give them a dreadful grade such as Terrible. As such, in a grading model, a blank vote will be interpreted as conscious decision to ask to be removed from the decision process. Think to referees who grade the papers of EC. The member might feel incompetent to express herself on a paper out of her field or might have a conflict of interest. In both cases, the referee might want to be excluded from grading that paper.

On the other hand, an absentee decided to ignore the chance to vote, and so in a sense entrusted the results to the rest of the voters. As such they can be considered to have voted for the outcome since they agree with what the other voters decided. Hence, the absentees will contribute to strengthen out trust in the election (because may be counted as if they voted for the outcome hence as if they follow the judgment of the colleagues who voted), and the blank vote will weaken that trust (because represents the desire to be removed from the election). In real life, examples where the distinction was made exist. In France the "Parti du vote blanc" (blank vote party) was created for those who wished for protest votes to count: if it receives most of the votes, one redo the election with new candidates. In other words, Blank is considered as a candidate who, when elected, all other candidates are rejected and the debate may continue with a new set of candidates.

\subsection{Blank Votes : BV}

Recall that blank votes are represented by $\otimes$ and absentees by $\circ$. In order to better describe the notion of blank vote, let us define the notion of completely ineligible voter. A voter is considered completely ineligible if there does not exist a candidate for which they may vote. Such a voter can be removed from the electorate without any impact on the outcome.

\begin{axiom}[Removing completely ineligible voters]
 A grading function $\varphi$ verifies the axiom if for any $i$ such that $\mathcal{C}_i = \emptyset$, we do not distinguish between $\mathcal{N}$ and $\mathcal{N}$ - \{i\}.
\end{axiom}

Many of the situations our model represents assume that proper grading of the candidates is time consuming. As such we can easily imagine that a lazy voter would wish to be considered completely ineligible.
Similarly, in certain situations a voter $i$ may wish to be considered ineligible to vote for a single candidate $J$. For example, if the voter is aware of a conflict of interest (as in peer review) it would make sense that they can submit a blank vote to inform the regulator that there has been a mistake and that they should not have been granted the right to vote for the candidate. The regulator should be able to adapt by determining what the outcome would had been if he did not granted $i$ the right to vote for $J$. It is immediate that a lazy voter can therefore request to be considered completely ineligible, in other words they can ask to be removed from the process. 

\begin{definition}[Blank votes : BV]
A function $\lambda: \mathcal{O}^{\mathcal{M} \times \mathcal{N}} \rightarrow {(\mathcal{B} \cup \{\emptyset\})}^\mathcal{M}$ respects blank votes (BV) if for all candidates $J$, for all voters $i$, for all $\textbf{v}$ where $v_i(J) =\otimes$, if $\textbf{w}=\textbf{v}[v_i(J):=\emptyset]$ then:
\[\lambda(\textbf{v})=\lambda(\textbf{w}).\]
\end{definition}

Interestingly enough, not being allowed to vote for a candidate is not the same as not being able to affect the candidate. As such unless a voter uses (BV) to make himself ineligible then depending on the situation he may still have an impact on the final outcome. In our current setting, it is perfectly possible for a voter that was not allowed to vote for a candidate to affect that candidate, thanks to the votes she provided to other candidates. This is somewhat counter-intuitive as some of the versions we can suggest for dealing with absentees will however completely remove their ability to affect the candidate.

\begin{theorem}[Blank vote characterization]\label{BV theo}
    A SP grading function $\varphi$ respects blank votes (BV) iff the phantom-mappings $\omega_{J,S}^T$ associated to $\varphi$ respect blank votes. (Proof \ref{BV theo proof}.)
\end{theorem}

It is important to realize that the intuitive implication is one direction only. Being able to replace a blank vote $\otimes$ with an ineligibility notification $\emptyset$ does not mean we can replace a ineligibility notification with a blank vote. A voter can demand to be removed from the election but he cannot request permission to join. In other words the notion of blank vote however implies that the regulator has to prepare the mechanism in such a way that he considers the possibility of removing additional vote rights.

In our setting when a voter provides a (BV) to candidate $J$ it usually represents his desire to let the regulator act as if he was not given eligibility to vote for $J$. In terms of phantom-proxies this therefore can describes sufficient trust in the ability of the proxy vote to properly represent them. 

\begin{proposition}\label{BV prop}
A phantom-proxy mechanism $\psi$ respects blank votes (BV) iff the $f_{J,i}$ proxy functions associated to $\psi$ respect blank votes. (Proof \ref{BV prop proof}.)
\end{proposition}

An example of a phantom-proxy mechanism that does not verify (BV) would be if for a given $i$ and $J$ we had $f_{i,J} = med(\{v_i(J) : J \in \mathcal{C'}_i\})$ when $v_i(J) =\otimes$ and $f_{i,J}=\emptyset$ when $v_i(J)=\emptyset$.

\subsection{Ignoring the absentees : SI}
One of our options for dealing with absentee voters is simply to ignore them completely. In other words, if a voter $i$ abstained for $J$ the regulator proceeds for $J$ as if the voter $i$ never existed from the perspective of $J$. This is quite different from (BV) where the outcome for $J$ can still take in account the way $i$ voted elsewhere. In particular when considering phantom-proxy mechanism since the voter did not provide a blank vote, we cannot trust that he is willing to allow a proxy to represent him and as such it make sense not to. Therefore since we do not know what to do with the absentee voter we might as well just remove him from the process.

\begin{definition}[Who is silent is ignored : SI]
A grading function $\varphi$ verifies the "who is silent is ignored" rule if for any candidate $J$, any voter $i$ and any $\alpha \in \mathcal{A}$ we have:
\[\forall \textbf{v} \in \mathcal{O}^{\mathcal{M}\times \mathcal{N}}, v_i(J)=\circ \Rightarrow  \varphi(\textbf{v})(J) = \varphi(\textbf{v}[\forall J; v_i(J):=\emptyset])(J).\]
\end{definition}

\begin{theorem}[Ignoring absentees characterization] \label{SI theo}
    A SP grading function $\varphi$ verifies the "who is silent is ignored" (SI) property iff we can select the phantom mappings $\omega_{J,S}^T$ such that they verify (SI). (Proof \ref{SI theo proof}.)
\end{theorem}

\begin{proposition}\label{SI prop}
A phantom-proxy mechanism $\psi$ verifies the "who is silent is ignored" (SI) property iff for all $i,J$ we have:
\[\forall v_i, v_i(J) = \circ \Rightarrow f_{i,J}(v_i)= \emptyset.\] (Proof \ref{SI prop proof}.)
\end{proposition}

\subsection{Silent consent rule : SC}
Our basic notion for an absentee is based on the expression "who stays silent consents" and depicts the fact that a voter that abstained might as well have voted the outcome. In other words, we should be able to act as if the outcome for $J$ when $i$ abstains for $J$ is the same as if $i$ voted the said outcome. While the intuition for the "silent consent rule" is relatively straightforward, it also implies that voter $i$ can vote for the outcome he would have obtained and this is not necessarily the case. As such let us start with the awkward definition in which we can only consent when the outcome was in $\mathcal{A}$ and its characterization before moving on to the more intuitive approach where we extend the set of possible grades from $\mathcal{A}$ to $\mathcal{B}$ so that our absentee voter can vote any possible outcome.

\begin{definition}[Who is silent consents rule : SC]
A grading function $\varphi$ verifies the "who is silent consents" rule if for any candidate $J$, any voter $i$ and any $\alpha \in \mathcal{A}$ we have:
\[\forall \textbf{v} \in \mathcal{O}^{\mathcal{M}\times \mathcal{N}}, v_i(J)=\circ \wedge \varphi(\textbf{v})(J) = \alpha \Rightarrow  \varphi(\textbf{v}[v_i(J):=\alpha])(J) = \varphi(\textbf{v})(J).\]
\end{definition}

\begin{lemma}[SC characterization]\label{SC theo lemma}
An SP $\varphi: \mathcal{O}^{\mathcal{M} \times \mathcal{N}} \rightarrow \mathcal{B}^\mathcal{M}$ function verifies the SC property if for all candidates $J$, all sets $\emptyset \subseteq S \subseteq T \subseteq \mathcal{D}_J$, all candidates $i \not \in T$ and all profiles $\textbf{v}$, the phantom-mappings verify that for every $\alpha \in \mathcal{A}$ and $\forall \textbf{v}$:
\[\omega_{J,S}^T(\textbf{v}_{-T}) \leq \alpha \leq \omega_{J,T}^T(\textbf{v}_{-T})\Rightarrow \omega_{J,S}^{T\cup \{i\}}(\textbf{v}_{-T \cup \{i\}})  \leq \alpha\]
\[\omega_{J,\emptyset}^T(\textbf{v}_{-T})\leq \alpha \leq \omega_{J,S}^T(\textbf{v}_{-T})   \Rightarrow \alpha \leq \omega_{J,S\cup \{i\}}^{T\cup \{i\}}(\textbf{v}_{-T \cup \{i\}}) \] (Proof \ref{SC theo lemma proof}.)
\end{lemma}

As observed above, an SP mechanism can always be extended so that $\mathcal{A}= \mathcal{B}$. In that case we have.

\begin{theorem}[Silent consent characterization]\label{SC theo}
When $\mathcal{A}= \mathcal{B}$, an SP grading function $\varphi$ verifies the silent consent rule (SC) if for all candidates $J$, all sets $\emptyset \subset S \subset T \subseteq \mathcal{D}_J$, all candidates $i \not \in T$ and all profiles $\textbf{v}$, the phantom-mappings verify:
\[\omega_{J,S}^{T\cup \{i\}}(\textbf{v}_{-T\cup \{i\}}) \leq \omega_{J,S}^{T}((\textbf{v}[v_i(J):=\circ]_{-T})\leq \omega_{J,S \cup \{i\}}^{T\cup \{i\}}(\textbf{v}_{-T\cup \{i\}}).\]
(Proof \ref{SC theo proof}.)
\end{theorem}

\begin{corollary}
    When $\mathcal{A}=\mathcal{B}$, An SP grading function $\varphi$ verifies SI and SC iff its associated phantom-mappings $\omega_{J,S}^T$ verify:
    \[\omega_{J,S}^{T\cup \{i\}}(\textbf{v}_{-T}) \leq \omega_{J,S}^{T}((\textbf{v}_{-T})_{-i}) = \omega_{J,S}^{T}((\textbf{v}_{-T})_{-i},v_i[J:=\circ])\leq \omega_{J,S \cup \{i\}}^{T\cup \{i\}}(\textbf{v}_{-T}).\]
\end{corollary}
 
 \begin{proposition}\label{SC prop}
 A phantom-proxy mechanism $\psi$ verifies the SC rule iff for all $p \in \mathcal{N}$ we have $g_J(p+1) \in \{g_J(p), g_J(p)+1\}$. (Proof \ref{SC prop proof}.)
 \end{proposition}
 
\subsection{Participation : P and FP}
Since our objective is to get voters to grade as much candidates as possible, we need be sure that they never benefit from an absentee vote. That is to say, if a voter has an ideal grade for a candidate, then he should not be able to get closer to that grade by abstaining for that candidate. This is the Participation property (also known as the no no-show paradox). We define two notions of Participation, one stronger than the other.



\subsubsection{Participation : P}

\begin{definition}[Participation : P]
A grading function $\varphi$ verifies the Participation property (P) if for any candidate $J$, any voter $i$ with $v_i(J) \in \mathcal{A}$:
\[\varphi(\textbf{v})(J) > v_i(J) \Rightarrow  \varphi(\textbf{v}[v_i(J):=\circ])(J) \geq \varphi(\textbf{v})(J)\]
\[\varphi(\textbf{v})(J) < v_i(J) \Rightarrow  \varphi(\textbf{v}[v_i(J):=\circ])(J) \leq \varphi(\textbf{v})(J)\]

\end{definition}

\begin{remark}\label{P remark}
Even if we are not SP, participation (P) implies that we verify the "silent consent" rule (SC). (Proof \ref{P remark proof}.)
\end{remark}

\begin{theorem}[Participation characterization] \label{P theo}
An SP grading function $\varphi$ verifies participation (P) if for candidates $J$ all sets $\emptyset \subset S \subset T \subseteq \mathcal{D}_J$, all candidates $i \not \in T$, all profiles $\textbf{v}$ the phantom-mappings verify:
\[\omega_{J,S}^{T\cup \{i\}}(\textbf{v}_{-T\cup \{i\}}) \leq \omega_{J,S}^{T}((\textbf{v}[v_i(J):=\circ])_{-T})\leq \omega_{J,S \cup \{i\}}^{T\cup \{i\}}(\textbf{v}_{-T\cup \{i\}}).\]
(Proof \ref{P theo proof}.)
\end{theorem}

\begin{theorem}[Relationship between P and SC]
When $\mathcal{A} = \mathcal{B}$, in the (SP) context, the participation and "silence consent" properties (SC) are equivalent.
\end{theorem}


\begin{proposition}\label{P prop}
A phantom-proxy mechanism verifies participation (P) iff it verifies the silent consent rule (SC). (Proof \ref{FP prop proof}.)
\end{proposition}

\subsubsection{Full Participation : FP}\ \\
In full participation (FP), voters cannot benefit from blank votes $\otimes$ either. 

\begin{definition}[Full Participation : FP]
A grading function $\varphi$ verifies participation (FP) if for any candidate $J$, any voter $i$ with $v_i(J) \in \mathcal{A}$ and any $\epsilon \in \{\circ,\otimes\}$:
\[\varphi(\textbf{v})(J) \geq v_i(J) \Rightarrow  \varphi(\textbf{v}[v_i(J):=\epsilon])(J) \geq \varphi(\textbf{v})(J)\]
\[\varphi(\textbf{v})(J) \leq v_i(J) \Rightarrow  \varphi(\textbf{v}[v_i(J):=\epsilon])(J) \leq \varphi(\textbf{v})(J)\]

\end{definition}

\begin{theorem}[Participation characterization]
An (SP) grading function $\varphi$ verifies participation (FP) if for all candidates $J$, all sets $S \subseteq T \subseteq \mathcal{D}_J$, and all voting profiles $\textbf{v}$, we have $\omega_{J,S}^{T}(\textbf{v}_{-T}) \leq \omega_{J,S}^{T \cup \{i\}}((\textbf{v}[v_i(J):=\otimes])_{-T})\leq \omega_{J,S}^{T}(\textbf{v}_{-T})$.
\end{theorem}

\begin{proposition}\label{FP prop}
Any phantom-proxy mechanism that verifies (P) verifies (FP). (Proof \ref{FP prop proof}.)
\end{proposition}

\subsection{Jury Determinism : JD}
We may also wish for the aggregated grade for a candidate not to be affected by the other candidates, but only depend on the votes of his assigned jury. 

\begin{definition}[Jury Determinism : JD]
A grading function $\varphi$ is jury determined (JD) iff for all \textbf{v} and \textbf{w}:
\[
 \varphi(w_1[J:=v_i(J) : \forall i \in \mathcal{N}])(J)
    = \varphi(\textbf{v})(J)\]
\end{definition}



\begin{lemma}\label{JD lemma}
An (SP) grading function $\varphi$ is jury determined (JD) iff all phantom-mappings $\omega_{J,S}^{T}$ only depend on the identity of absentees for $J$ and $\mathcal{D}_J$. (Proof \ref{JD lemma proof}.)
\end{lemma}

It is unsurprising to find that JD and phantom-proxy mechanisms do not provide a large class of mechanisms. The purpose of JD is to reduce the impact of non-eligible voters on the outcome whereas the class of phantom-proxy mechanism is intended as a means to provide grades to represent the non-eligible voters. Intuitively, the 2 conflict, as the following propositions shows. They show how limited phantom-proxy mechanism become in a JD setting where there is almost no information available for the proxy to learn how to imitate their voter.

\begin{proposition}\label{JD prop}
A phantom-proxy mechanism $\psi$ verifies jury determinism (JD) iff for all $f_{i,J}$ functions, there is a function $\tilde{f_{i,J}} : \mathcal{E} \rightarrow \mathcal{B}$ such that:
\[\forall v_i, f_{i,J}(v_i) = \tilde{f_{i,J}}(v_i(J)).\] (Proof \ref{JD prop proof}.)
\end{proposition}

\subsection{Strong strategy-proofness : Strong SP}

This is SP except that we now consider that you might have an opinion for a candidate you were not asked to vote for. If this happens then you should not be able to impact its outcome in a way that is make their grade closer to your opinion.

\begin{definition}[Strong strategy-proofness : Strong SP]
A grading function $\varphi$ is strongly strategy-proof if for any voter $i$, for any candidate $J$ and $\alpha \in \mathcal{A}$ and any $\textbf{v}$ and $w_i$ such that if $v_i(J) \in \mathcal{A}$ then $v_i(J)=\alpha$ and $\mathcal{C}_i(\textbf{w})= \mathcal{C}_i(\textbf{v})$:
\[\varphi(\textbf{v})(J) > \alpha \Rightarrow  \varphi(\textbf{v}[v_i:=w_i])(J) \geq \varphi(\textbf{v})(J)\]
\[\varphi(\textbf{v})(J) < \alpha \Rightarrow  \varphi(\textbf{v}[v_i:=w_i])(J) \leq \varphi(\textbf{v})(J).\]
\end{definition}

Strong SP states that even if an expert did not give a grade, they cannot manipulate the outcome, even if they were not offered the option to grade the candidate. This almost implies jury determinism.

\begin{lemma}
    Strong SP is equivalent to (SP,FP,JD).
\end{lemma}

\section{Additional properties and their characterisations}\label{section prop}
In the following sections we describe different properties (or axioms) that are considered desirable.

\subsection{Unanimity : U}

Unanimity usually corresponds to the situation where all voters agree on something. It is quite natural that in such an instance we want the outcome to also agree. In our setting, we desire a stronger notion of unanimity. When all voters that graded a candidate agree about the grade, then the outcome for that candidate should be that grade.

\begin{definition}[Unanimity : U]
A grading function $\varphi$ is unanimous iff:
\[\forall \textbf{v}, \forall \alpha \in \mathcal{A}, \forall i, v_i(J) \in \{\alpha\} \cup \mathcal{E} \wedge \exists i, v_i(J) = \alpha \Rightarrow \varphi(\textbf{v})(J) = \alpha. \]
\end{definition}


\begin{theorem}\label{U theo}
A SP grading function $\varphi$ is unanimous iff the phantom-mappings satisfy:
\[\forall J, \emptyset \subset T \subseteq \mathcal{D''}_J, \omega_{J,\emptyset}^T \leq \inf \mathcal{A} \wedge \omega_{J,T}^T\geq\sup\mathcal{A}.\]
(Proof \ref{U theo proof}.)
\end{theorem}

Interestingly enough, we can see that if $\varphi$ is unanimous then so long as $J$ received at least one grade the outcome for $J$ is always in between $\inf \mathcal{A}$ and $\sup\mathcal{A}$ included. 

It is known that in single-peaked settings, Unanimity is equivalent of Pareto Optimality \cite{Weymark:2011}.
It is also true in our setting (see appendix \ref{Pareto}). This means that (SP,U) are equivalent to: for a fixed candidate $J$, no grade $\beta \in \mathcal{B}$ could be closer to the grade a voter provided without being further away from the grade another voter provided. 

Unanimity is not a natural property to obtain when considering phantom-proxy mechanisms. In order to have a unanimous phantom-proxy mechanism we would need to be sure that the number of proxy votes is never greater than the number of grades provided. In the instance where only one grade was provided we therefore cannot have any proxy votes or that the $g_J$ function always selects the right outcome.

\begin{proposition}
A phantom-proxy $\psi$ mechanism is unanimous (U) iff $\mathcal{F}_J=\emptyset$
\end{proposition}

\subsection{Neutrality : N, SN and F}
A decision maker is neutral if it is considered a fair independent judge that treats all candidates equality. The decision they make should only depend on the information provided for the decision process and not any personal preferences regarding the candidates. 
We therefore expect that if two candidates swapped their names (without telling the decision maker) everything else equal, their positions will have swapped in the outcome (everything else equal). In order to promote fairness we therefore wish to introduce a notion of neutrality. However the notion of neutrality conflicts with restricted voting. How can we expect the regulator to treat candidates the same when they have already created non-equal treatment with the voting rights. Our first notion of neutrality (N) is therefore very restrictive. The regulator can only be undiscriminating between two candidates if he gave them the same set of voters. The stronger version of neutrality (SN) is more permissive, as voting rights are not taking into account. We also introduce fairness (F), a weak of neutrality that is specific to phantom-proxy mechanism

\subsubsection{Neutrality : N}
In the case of restricted neutrality (or neutrality), we can only swap the votes for 2 candidates if they have the same rights to vote. 
We therefore expect that the notion of restricted neutrality is linked to the ability to bound a candidate to the set of voters that voted for it. 

\begin{definition}[Neutrality: N]
    A function $\lambda : \mathcal{O}^{\mathcal{M} \times \mathcal{N}} \rightarrow {(\mathcal{B} \cup \{\emptyset\})}^\mathcal{M}$ is neutral (N) if any two candidates $I$ and $J$ such that $\mathcal{D}_I=\mathcal{D}_J$:
    \[
 \varphi(\textbf{v}[v_i[I:=v_i(J);J:=v_i(I)] : \forall i \in \mathcal{D}_I])
    = \varphi(\textbf{v})[I:= \varphi(\textbf{v})(J); J := \varphi(\textbf{v})(I)]
\]
\end{definition}

The following theorem is therefore relatively expected.

\begin{theorem}\label{N theo}
A SP grading function $\varphi$ verifies N if there exists a neutral $\omega_{U,S}^{T}$ functions where $T \subseteq U \subseteq \mathcal{N}$ such that if $\mathcal{D}_J = U$ then $\omega_{J,S}^{\mathcal{D''}_J}(\textbf{v}_{-\mathcal{D''}_J})=\omega_{U,S}^{\mathcal{D''}_J}(\textbf{v}_{-\mathcal{D''}_J})$. (Proof \ref{N theo proof}.)
\end{theorem}

\begin{proposition}\label{N prop}
Let $\bigcup \mathcal{M}_S =\mathcal{M}$ be the partition of $\mathcal{M}$ defined by $J \in \mathcal{M}_S$ if $\mathcal{D}_J = S$. A phantom-proxy mechanism $\psi$ is N iff for all $i$, and all $S\subseteq \mathcal{N}$ we have an neutral function $f_{i,S}$ and a function $g_S$ such that $f_{i,J}=f_{i,S}$ and $g_J=g_S$ if $J \in \mathcal{M}_S $. (Proof \ref{N prop proof}.)
\end{proposition}

\subsubsection{Strong Neutrality : SN}
In strong neutrality we can swap the rights to vote. 
It implies that the regulator should not care who has which rights. For example, the allocation of rights to vote may be random or once the rights to vote was allocated the regulator forgot about them. Alternatively, the regulator may simply create the mechanism without defining the rights to vote yet and intends to use the (BV) property to determine all values once they decided on the rights to vote.

\begin{definition}[Strong neutrality : SN]
A method $\varphi$ is strong neutral (SN) if for any two candidates $I$ and $J$:
\[
 \varphi(\textbf{v}[v_i[I:=v_i(J);J:=v_i(I)] : \forall i])
    = \varphi(\textbf{v})[I:= \varphi(\textbf{v})(J); J := \varphi(\textbf{v})(I)]
\]

\end{definition}

\begin{theorem}[Strong Neutrality characterization]\label{SN theo}
A (SP) grading method $\varphi$ verifies (SN) iff there exists a strong neutral $\omega_{S}^{T}$ function such that for all $\mathcal{J}$ we have $\omega_{S,J}^{T}=\omega_{S}^T$. (Proof \ref{SN theo proof}.)
\end{theorem} 

In other words, a grading function is strongly neutral if we use the same phantom-mappings $\omega_{S}^T$ for all the candidates. This was to be expected. The easiest way to ensure that the outcome does not depend on the identities is to ensure that no step in process depends on their identities.

\begin{proposition}\label{SN prop}
A (BV) phantom-proxy mechanism $\psi$ verifies (SN) iff there is a strongly neutral function $f_i$ such that $f_{i,J}=f_i$ for all $J$ and there is a function $g$ such that for all $J$ we have $g=g_J$. (Proof \ref{SN prop proof}.)
\end{proposition}

\subsubsection{Fairness : F}
In the specific case of phantom-proxy mechanism we have an additional notion of neutrality. The intuition behind the phantom-proxy mechanisms is that voters can be represented by the grade they submitted or by some proxy vote. Once the voting pool is selected, the mechanism decides who the winner is. A phantom-proxy mechanism is considered fair is 2 different candidates that provided the same voting pool get the same outcome.

\begin{definition}[Fairness : F]
A phantom-proxy mechanism $\psi$ is fair iff for any two candidates $I$ and $J$:
\[\textbf{v}(J) \cup \mathcal{F}_J(\textbf{v}) = \textbf{v}(J) \cup \mathcal{F}_J(\textbf{v}) \Rightarrow \psi(\textbf{v})_I=\psi(\textbf{v})_J\]
\end{definition}

\begin{theorem}
    A phantom-proxy $\psi$ is fair iff there exists $g$ such that for all $J$, $g_J=g$.
\end{theorem}

\subsection{Anonymity : A and SA}
In the classical sitting, a voting method is anonymous if the identity of a voter does not impact the outcome of the election. In other words any 2 voters can swap their ballots without it impacting the outcome. In our model with rights, this notion is troublesome. The identity of the voter impacts which candidates he has the right to vote for. It would therefore be natural that two voters without the same rights to vote cannot swap their ballots. We are therefore left with 2 different anonymity notions. The one were we are restricted and can only swap ballots if the voters have the same rights to vote and the one were any 2 ballots can be swapped as if we are ignoring the rights to vote. 





\subsubsection{Anonymity : A}
Just like with neutrality, we 
consider that 2 voters can only swap their ballots if they have the same rights to vote. 
\begin{definition}[Anonymity]
    A function $\lambda : \mathcal{O}^{\mathcal{M} \times \mathcal{N}} \rightarrow {(\mathcal{B} \cup \{\emptyset\})}^\mathcal{M}$ is anonymous (A) if for any $\textbf{v}$ and for any $\textbf{w}$ obtained from $\textbf{v}$ by switching the ballots of 2 players that had the same rights to vote we have:
    \[\varphi(\textbf{v}) = \varphi(\textbf{w}).\]
\end{definition}

Let $\bigcup_{M\subseteq \mathcal{M}} \mathcal{N}_M$ be the partition of $\mathcal{N}$ such that if $\mathcal{C}_i = M$ then $i \in \mathcal{N}_M$. This partition provides us the sets of voters that can swap their votes. We say that $T$ and $U$ have the same partition-cardinal if $\forall M \subseteq \mathcal{M}, \#(\mathcal{N}_M\cap T)=\#(\mathcal{N}_M\cap T)$.

\begin{theorem}[Anonymous characterization]\label{A theo}
    A SP grading function $\varphi $ is anonymous iff all phantom-mapping are anonymous and if for all $J$, if $T$ and $T'$ have the same cardinal-partition, and if $S\subseteq T$ and $S'\subseteq T'$ have the same cardinal-partition then $\omega_S^T=\omega_{S'}^{T'}$ (Proof \ref{A theo proof}).
\end{theorem}

\begin{proposition}\label{A prop}
A (BV) phantom-proxy mechanism $\psi$ verifies anonymity (A) iff for all $J$, for all $M \subseteq \mathcal{M}$ there is a $f_{M,J}$ such that if $\mathcal{C}_i=M$ then $f_{M,J}=f_{i,J}$. (Proof \ref{A prop proof}.)
\end{proposition}

Here is an example with 3 voters $\{x,y,z\}$ and 2 candidates $\{I,J\}$. Let $\mathcal{C}_x=\{I,J\}$ and $\mathcal{C}_y=\mathcal{C}_z =\{I\}$. Then if $f_{I,y}=f_{I,z}$ and $f_{y,I}=f_{z,I}$, then the mechanism is anonymous. Else it isn't.

\subsubsection{Strong Anonymity : SA}
An intuitive anonymity situation is strong anonymity. This represents situations where the regulator treats all voters the same regardless of their voting rights (for example if there are many voters and they were assigned at random, like in LaPrimaire.org). 

\begin{theorem}[Strong Anonymous characterization]\label{SA theo}
An (SP) grading function $\varphi : \mathcal{O}^{\mathcal{M} \times \mathcal{N}} \rightarrow (\mathcal{B} \cup \{\emptyset\})^\mathcal{M}$ is strongly anonymous iff for each $J$ there exists $ \dfrac{\#\mathcal{D}_J(\#\mathcal{D}_J+1)}{2}$ strongly anonymous phantom-mappings $\omega_{J,k}^{d} : \mathcal{O}^{\mathcal{M}} \times \mathcal{N} \Rightarrow \mathcal{B}$ such that for all $0 \leq k < d \leq \#\mathcal{D}_J$ we have  $\omega_{J,k}^{d} \leq \omega_{J,k+1}^{d}$ and:
\[\forall \textbf{v}, \varphi(\textbf{v})(J) = med(\textbf{v}(J),\omega_{J,0}^{\#\mathcal{D''}_J}(\textbf{v}_{-\mathcal{D''}_J}),\dots, \omega_{J,\#\mathcal{D''}_J}^{\#\mathcal{D''}_J}(\textbf{v}_{-\mathcal{D''}_J})).\] (Proof \ref{SA theo proof}.)
\end{theorem}

This is perhaps the closest formula we have to \cite{M1980} median in the anonymous case. This is largely due to the fact that just like with the proof of theorem \ref{SP theo} we can use the \cite{Moulin} formula as the first step of the proof of this theorem.

\begin{remark}
        As in \cite{M1980}, to obtain a unanimous rule, we just need to remove $\omega_{J,0}^d$ and $\omega_{J,d}^d$.
\end{remark}

\begin{proposition}\label{SA prop}
A (BV) phantom-proxy mechanism $\psi $ verifies anonymity iff for all $J$ there is a function $J$ such that for all $i$, $f_{i,J}=f_J $.(Proof \ref{SA prop proof}.)
\end{proposition}

\subsection{Consistency : OC and IC}
Let us consider a situation where the regulator performed his survey in 2 different towns and obtained the same grade for a candidate in both. It is natural to consider that if the regulator had performed a single survey on the total population of both towns the candidate should have gotten that same grade. This is the idea of consistency (in grading). We omit ``grading'' as we will not discuss consistence in ranking because of lack of space. Consistency was first introduced and studied in the context of social choice (resp. welfare) functions by \cite{young} and\cite{Smith}.

In this section we consider two versions of consistency. In the first all voters are assigned to one of the two sets we wish to merge. This is the most common version of consistency in the literature. In the second, only the two groups of voters that provided grades to be disjoint.

\subsubsection{Outer Consistency : OC}
In outer consistency, the two sets that we wish to merge are considered completely disjoint. As such, a voter is considered ineligible when another set is being studied and completely represented when it is part of the set.

 \begin{definition}[Outer consistency :OC]
A grading function $\varphi$ is outer-consistent (OC) iff for all partitions $\mathcal{N}_1 \cup \mathcal{N}_2=\mathcal{N}$ of the set of voters and all voting profile $\textbf{t}$ and all candidates $J$, if $\textbf{v}$ and $\textbf{w}$ are the voting profiles defined as $\textbf{v}= \textbf{t}_{-\mathcal{N}_1}$ and $\textbf{w}= \textbf{t}_{-\mathcal{N}_2}$ then we have:
\[\varphi(\textbf{v})(J) =\varphi(\textbf{w})(J) \Rightarrow \varphi(\textbf{v})(J) = \varphi(\textbf{t})(J) \]
\end{definition}


Note that the previous definition implies almost (BV) to a certain extent.

\begin{theorem}\label{OC theo}
An (SP,BV) grading function $\varphi$ is verifies consistency (OC) iff for all $J$, for all $\textbf{t}$ and for all partitions $\mathcal{N}_1 \cup \mathcal{N}_2=\mathcal{N}$, if 
$\textbf{v} =\textbf{t}_{-\mathcal{N}_1}$ and $\textbf{v} =\textbf{t}_{-\mathcal{N}_2}$ we have that for all $S \subseteq \mathcal{D''}_J(\textbf{v})=T$ and $S' \subseteq \mathcal{D''}_J(\textbf{w})=T'$ we must verify:

\begin{enumerate}
    \item \[\omega_{J,S}^T(\textbf{v}_{-T})= \omega_{J,S}^{T'}(\textbf{w}_{-T'}) \Rightarrow \omega_{J,S\cup S'}^{T\cup T'}(\textbf{t}_{-(T\cup T')}) \]
    \item \begin{align*}
   \max(\omega_{J,\emptyset}^T(\textbf{v}_{-T}), \omega_{J,\emptyset}^{T'}(\textbf{w}_{-T'})) \leq & \alpha \leq \min(\omega_{J,S}^T(\textbf{v}_{-T}), \omega_{J,S}^{T'}(\textbf{w}_{-T'})) \\ 
   & \Rightarrow \\
      \omega_{J,\emptyset}^{T\cup T'}(\textbf{t}_{-(T\cup T')}) \leq & \alpha \leq \omega_{J,S\cup S'}^{T\cup T'}(\textbf{t}_{-(T\cup T')}) 
\end{align*}
    \item \begin{align*}
   \max(\omega_{J,S}^T(\textbf{v}_{-T}), \omega_{J,S'}^{T'}(\textbf{w}_{-T'})) \leq & \alpha \leq \min(\omega_{J,T}^T(\textbf{v}_{-T}), \omega_{J,T}^{T'}(\textbf{w}_{-T'})) \\ 
   & \Rightarrow \\
      \omega_{J,S\cup S'}^{T\cup T'}(\textbf{t}_{-(T\cup T')}) \leq & \alpha \leq \omega_{J,T\cup T'}^{T\cup T'}(\textbf{t}_{-(T\cup T')}) 
\end{align*}
\end{enumerate} (Proof \ref{OC theo proof}.)
\end{theorem}

\begin{corollary}
When $\mathcal{A}=\mathcal{B}$, a (SP,BV) grading function $\varphi$ verifies consistency (OC) iff for all $J$, for all $\textbf{t}$ and for all partitions $\mathcal{N}_1 \cup \mathcal{N}_2=\mathcal{N}$, if 
$\textbf{v} =\textbf{t}_{-\mathcal{N}_1}$ and $\textbf{v} =\textbf{t}_{-\mathcal{N}_2}$ we have that for all $S \subseteq \mathcal{D''}_J(\textbf{v})=T$ and $S' \subseteq \mathcal{D''}_J(\textbf{w})=T'$ we must verify:

\[\min(\omega_{J,S}^T(\textbf{v}_{-T}), \omega_{J,S}^{T'}(\textbf{w}_{-T'})) \leq \omega_{J,S\cup S'}^{T\cup T'}(\textbf{z}_{-(T\cup T')}) \leq \max(\omega_{J,S}^T(\textbf{v}_{-T}), \omega_{J,S}^{T'}(\textbf{w}_{-T'}))\]

\end{corollary}

\begin{proposition}\label{OC prop}
A (BV) phantom-proxy $\psi$ is outer-consistent iff for all $J$ and $\forall k,k'$ we have $g_J(k+k') \in \{g_J(k) + g_J(k') -1; g_J(k) + g_J(k')\}$. (Proof \ref{OC prop proof}.)
\end{proposition}

\begin{remark}
 A (BV,OC) phantom-proxy mechanism $\psi$ satisfies the participation (P).
\end{remark}

\subsubsection{Inner Consistency : IC}
Inner consistency is a stronger notion of consistency where, if we have two incomplete voting profiles for the same set of voters, 
if both voting profiles agreed on the grade for $J$ then we would like for the outcome to agree as well. 

 \begin{definition}[Inner consistency :IC]
Let $\textbf{t}$ be the voting profile where all voters were allowed to vote. Let $\textbf{v}$ and $\textbf{w}$ be voting profiles obtained from $\textbf{t}$ by removing the right to vote of some experts for some candidates. A grading function $\varphi$ is inner-consistent (IC) iff for any such $\textbf{t}, \textbf{v}, \textbf{w}$, if for candidate $J$, $\textbf{v}$ and $\textbf{w}$ have disjoints sets of voters with the right to vote for $J$, then:
\[\varphi(\textbf{v})(J) =\varphi(\textbf{w})(J) \Rightarrow \varphi(\textbf{v})(J) = \varphi(merge(\textbf{v},\textbf{w}))(J) \]
where merge is the  function defined by : if for any $(i,J)$, $v_i(J) \neq \emptyset \Rightarrow merge(\textbf{v},\textbf{w})_i(J)=v_i(J)$. 
\end{definition}

\begin{theorem}\label{IC theo}
An (SP,BV) grading function $\varphi$ verifies inner-consistent (IC) if for all $J$, all voting profiles $\textbf{v}$ and $\textbf{w}$ obtained from a same $\textbf{t}$ by removing votes such that $\mathcal{D''}_J(\textbf{v})$ and $\mathcal{D''}_J(\textbf{w})$ are disjoint, we have that for all $S \subseteq \mathcal{D''}_J(\textbf{v})=T$ and $S' \subseteq \mathcal{D''}_J(\textbf{w})=T'$ $\textbf{z}=merge(\textbf{v},\textbf{w})$ must verify:

\begin{enumerate}
    \item \[\omega_{J,S}^T(\textbf{v}_{-T})= \omega_{J,S}^{T'}(\textbf{w}_{-T'}) \Rightarrow \omega_{J,S\cup S'}^{T\cup T'}(\textbf{z}_{-(T\cup T')}) \]
    \item \begin{align*}
   \max(\omega_{J,\emptyset}^T(\textbf{v}_{-T}), \omega_{J,\emptyset}^{T'}(\textbf{w}_{-T'})) \leq & \alpha \leq \min(\omega_{J,S}^T(\textbf{v}_{-T}), \omega_{J,S}^{T'}(\textbf{w}_{-T'})) \\ 
   & \Rightarrow \\
      \omega_{J,\emptyset}^{T\cup T'}(\textbf{z}_{-(T\cup T')}) \leq & \alpha \leq \omega_{J,S\cup S'}^{T\cup T'}(\textbf{z}_{-(T\cup T')}) 
\end{align*}
    \item \begin{align*}
   \max(\omega_{J,S}^T(\textbf{v}_{-T}), \omega_{J,S'}^{T'}(\textbf{w}_{-T'})) \leq & \alpha \leq \min(\omega_{J,T}^T(\textbf{v}_{-T}), \omega_{J,T}^{T'}(\textbf{w}_{-T'})) \\ 
   & \Rightarrow \\
      \omega_{J,S\cup S'}^{T\cup T'}(\textbf{z}_{-(T\cup T')}) \leq & \alpha \leq \omega_{J,T\cup T'}^{T\cup T'}(\textbf{z}_{-(T\cup T')}) 
\end{align*}
\end{enumerate} (Proof \ref{IC theo proof}.)
\end{theorem}

\begin{corollary}
When $\mathcal{A}=\mathcal{B}$, an (SP,BV) grading function $\varphi$ verifies inner consistency (IC) if for all $J$, all voting profiles $\textbf{v}$ and $\textbf{w}$ obtained from a same $\textbf{t}$ by removing votes such that $\mathcal{D''}_J(\textbf{v})$ and $\mathcal{D''}_J(\textbf{w})$ are disjoint, we have that for all $S \subseteq \mathcal{D''}_J(\textbf{v})=T$ and $S' \subseteq \mathcal{D''}_J(\textbf{w})=T'$, $\textbf{z}=merge(\textbf{v},\textbf{w})$ must verify:

\[\min(\omega_{J,S}^T(\textbf{v}_{-T}), \omega_{J,S}^{T'}(\textbf{w}_{-T'})) \leq \omega_{J,S\cup S'}^{T\cup T'}(\textbf{z}_{-(T\cup T')}) \leq \max(\omega_{J,S}^T(\textbf{v}_{-T}), \omega_{J,S'}^{T'}(\textbf{w}_{-T'}))\]
\end{corollary}

\section{Ranking: an iterative tie-breaking for the proxy-mechanisms}
The mechanisms that we have considered so far allow us to provide grades which we can use to rank. The main issue with this grading method is that ties are expected and frequent, especially if we use a grading function where the output is one of the inputs and the input space $\mathcal{A}$ is small. For example, if we have only six possible grades as inputs grades such as \{Excellent, Good, Acceptable, Weak, Bad, Terrible\} then since the proxi-mechanisms will return one of these grades, for seven or more candidates, two will be in a tie. It is thus quite natural that if the aggregate grades have any meaning such as merit, two candidates that do not have the same final grades, we must rank them according to the order induced from their respective final grades. Our only issue is therefore how to break the ties. In this section we introduce an intuitive tie-breaking rule for the phantom-proxy mechanisms, extending the one used in \cite{BL02007,BL2011} for majority judgment. 

Importantly, as \cite{Gibbard} and \cite{Satterthwaite} theorems forbid the existence of strategy-proof mechanisms when the objective is to rank, one cannot expect our ranking methods to be strategy-proof in a ranking context. However, as majority judgment, proxy ranking methods can be proved to be partially strategy-proof in ranking (see \cite{BL2011}, section 13.1).


\subsection{Ranking when the voting pools have the same size}\label{Tie-break same size}

Let $\psi$ be a (F) phantom-proxy grading function, which is by definition, defined for a variable electorate. For any $J$, let $\mathcal{V}(J):=\textbf{v}(J) \cup \mathcal{F}_J(\textbf{v})$ be the voting pool for $J$. We therefore have that $\psi(\textbf{v})(J)=\mu_{g(\#\mathcal{V}(J))}(\mathcal{V}(J))$. Let $I$ and $J$ be two candidates we wish to rank using a phantom-proxy $\psi$ (not necessarily (BV) or (OC)) and suppose in this section (to facilitate the understanding) that: 


\[\forall I,J\in \mathcal{M},\#(\textbf{v}(I) \cup \mathcal{F}_I(\textbf{v})) = \#(\textbf{v}(J) \cup \mathcal{F}_J(\textbf{v}))\]
and

\[\forall i\in \mathcal{N}, v_i(J) \in \mathcal{A} \vee f_{i,J}(v_i) \in \mathcal{B} \Rightarrow \forall I, v_i(I) \in \mathcal{A} \vee f_{i,I}(v_i) \in \mathcal{B}\]

Having made these assumptions (to be relaxed in the next section), let us explain how the tie breaking-rule associated to a proxy $\psi$ rank candidate $J$ compared to $I$ and call this order $<_{\psi}$.

It is natural to assume that if  $\psi(\textbf{v})(J) < \psi(\textbf{v})(I)$ (the final grade of $J$ according to $\psi$ is strictly smaller than that of $I$, then $J <_{\psi} I$.

Therefore if we were unable to distinguish $I$ and $J$ according to $\psi$, that means that there is an $\alpha$ such that $\psi(\textbf{v})(J) = \psi(\textbf{v})(I) =\alpha$ and:
\[\alpha \in \mathcal{V}(I) \cup \mathcal{V}(J).\]

 If there is a voter $i$ that voted $\alpha$ for both candidates, he is satisfied by the outcome (both have $i$'s grade as final grade), and $i$ is satisfied with either tie-breaking choice between $I$ and $J$. It therefore makes sense to remove $i$ (or equivalently $\alpha$) from the ranking process to allow the other voters (in particular those with an incentive to participate in the tie-breaking process) to decide how the tie must be broken. If we accept this logic, to compare $I$ and $J$ we should compare $\psi(\textbf{v}[\forall J,v_i(J):=\emptyset])(I)$ and $\psi(\textbf{v}[\forall J,v_i(J):=\emptyset])(J)$. If they can distinguish the two candidate we are done, if not we repeat the process attractively until we can distinguish the two candidates or otherwise, we declare them equally competent (which happens only if they have identical voting pools, very unlikely in practice). 
 
 Following this logic, we generate a ranking function  $<_{\psi}$, better described if we associate to each candidate $J$ a voting range $R_J^{\psi} :  \mathcal{O}^{\mathcal{M} \times \mathcal{N}} \rightarrow (\mathcal{B})^\mathcal{N}$ generated by the following algorithm:

 \begin{enumerate}
     \item $R_J^{\psi} := []$; $S =\emptyset$;
     \item If $\psi(\textbf{v}[\forall J,\forall i \in S, v_i(J) = \emptyset])(J) = \emptyset$; Return $R_J^{\psi}$, END \label{R_G start loop}
     \item Else $\alpha :=\psi(\textbf{v}[\forall J,\forall i \in S, v_i(J) = \emptyset])(J)$ $R_J^{\psi} := (R_J^{\psi},\alpha)$;
     \item Find $i \in \mathcal{N}-S$ such that $v_i(J)=\alpha$ or $f_{i,J}(v_i) = \alpha$; $S: = S \cup \{i\}.$
     \item Go back to (\ref{R_G start loop})
\end{enumerate}

As $\psi$ if fixed we will drop in the sequel its index just use the notation $R_J$. Because the votes associated to each voter in the voting pool for a candidate $J$ is independent from the ballots of all other voters, we conclude that:
 
\begin{proposition}\label{R_J deter}
     The voting range associated to $J$ is well defined. (Proof \ref{R_J deter proof}.)
\end{proposition}

 \begin{definition}[$\psi$-Ordering]
 We declare $I<_{\psi} J$ at $\textbf{v}$ if and only if $R_I(\textbf{v}) < R_J(\textbf{v})$ for the lexicographic order.
 \end{definition}

When all voters are eligible and the voters input a grade to all the candidates, two well-known ranking functions are the leximin and leximax ordering \cite{Moulin}. Majority judgment \cite{BL02007}, which ranks iteratively according to the smallest median, is another.

\begin{corollary}\label{ranking corol}
    The ranking function $<_{\psi}$ defines a total (hence complete and transitive) order where a tie between two candidates happens only when their voting pools are identical. (Proof \ref{ranking corol}.)
\end{corollary}

 Observe that, to rank $I$ and $J$ using the voting ranges defined by $\psi$, the only requirement is that the voting pool sizes for $I$ and $J$ are equal. The next section extends the construction to non-equal sizes, which is very important in many applications.
 
\subsection{Ranking when voting pools have different sizes (under OC)}


When the voting pool size for two candidates are not initially equal, we will duplicate their respective voting pools until with get the same size. Operationally, this is equivalent to pass to a continuum electorate and normalize so that each voting pools size is $100\%$. This is how LaPrimaire.org and Paris participatory budget uses MJ to compare candidates with juries of different sizes. 

\begin{definition}[Duplicate voters]
For a phantom-proxy $\psi$, a voter $j$ is the duplicate of $i$ at the voting profile $\textbf{v}$ iff $v_i=v_j$ and $f_{i,J}=f_{j,J}$ for all $J$.
\end{definition}

Now let us make the link with the range function defined in the previous section. When $\psi$ is a proxy and $\textbf{v}$ is voting profile, $R_J(\textbf{v})$ is a vector $(R_J(\textbf{v})(1), ...,R_J(\textbf{v})(d))$ for some $d$. What happens to this vector if we duplicate the electorate $k$ times?  In general, anything can happen. However, if $\psi$ satisfies the outer consistent (OC) property, the next proposition provides a nice answer. Assuming (OC) makes sense because (OC) is a desirable property for variable electorates. Recall that it implies in particular the Participation (P), the Blank Vote (BV) and the Silent Ignored (SI) properties. 

\begin{proposition}\label{duplicate prop}
    When a phantom-proxy mechanism $\psi$ verifies (OC), for all voting profile $\textbf{v}$, if $\textbf{w}$ is obtained by duplicating all the voters $k$ times, then for $\textbf{w}$, the order in which we select for the first time a duplicate representing a new voter is the same as the order in which the voters where selected for $\textbf{v}$. (Proof \ref{duplicate proof})
\end{proposition}

    


\begin{corollary}\label{duplicate corol}
     When a voting phantom-proxy mechanism $\psi$ verifies (OC,F) and the voting profile $\textbf{v}$ provides us $I<_\psi J$ then any voting profile $\textbf{w}$ obtained by duplicating $k$ times all voters of $\textbf{v}$ also verifies $I<_\psi J$. (Proof \ref{duplicate corol proof})
\end{corollary}


Consequently, all candidates that could be ranked (using the previous section) because their voting pools have equal size, can still be ranked after duplication, and will keep the same rank as before.  We will show now how two candidates with non-equal pools can be ranked under (OC).

If for a given $\textbf{v}$, $I$ and $J$ cannot be ranked by using $\psi(\textbf{v})(I) < \psi(\textbf{v})(J)$ then we have that $\psi(\textbf{v})(I) = \psi(\textbf{v})(J)$ and that the voting pool size for $I$ and $J$ are non-equal. 


Let $\mathcal{N}_I$ and $\mathcal{N}_J$ be the set of voters obtained by duplicating respectfully $n_J$ and $n_I$ times $\mathcal{N}$. Let $\textbf{v}$ be the initial voting profile and let $\textbf{v}^I$ and $\textbf{v}^J$ be the voting profiles for $\mathcal{N}_I$ and  $\mathcal{N}_J$ respectively. Consequently, the voting pool of $I$ for $\textbf{v}^I$ has now the same size as the voting pool of $J$ for $\textbf{v}^J$. It follows that the voting range $R_I(\textbf{v}^I)$ can be compared to the voting range $R_J(\textbf{v}^J)$, as such we can order $I$ and $J$ according to $\psi$. Since $R_I(\textbf{v}^I)$ represents $R_I(\textbf{v})$ and $R_J(\textbf{v}^J)$ represents $R_J(\textbf{v})$ we can therefore, without loss of generality, claim that $R_I(\textbf{v}) < R_J(\textbf{v})$ iff $R_I(\textbf{v}^I) < R_J(\textbf{v}^J)$. 

\begin{proposition}
    The ranking method just described, we note $>_{\psi}$, induces a well defined total order.
\end{proposition}

Interestingly, we can view our ranking function as a grading function where the output space is the range space (typically $\cup_n \mathcal{B}^n$), which is a much richer output space than $\mathcal{B}$. In this richer space, ties are extremely rare. Also, if we assume that the preferred range outcome on a candidate $J$ for a voter $i$ that graded $J$ is $(v_i(J))$ and that a voter has a single peaked preference over that totally ordered space (for any two ranges $R_1$ and $R_2$ of same size, if $R_1$ and $R_2$ have the same prefix then use the SP on the first different value with $v_i(J)$ as your peak), then, thanks to our construction, our richer grading function is strategy-proof.

\begin{theorem}\label{newSP theo}
    The grading function induced by $>_{\psi}$ on the richer output space is strategy-proof. (Proof \ref{newSP theo proof})
\end{theorem}





\section{Extensions and conclusion}

The model is so rich that what is left to do exceed what we did. For example: we only characterized the two most natural anonymity properties, while there are many other interesting permutations that transform a $\textbf{v}$ into a $\textbf{w}$ everything else equal such as:
\begin{itemize}

  \item For any $i,j \in \mathcal{D''}_J$ : $w_i(J) = v_j(J)$ and $w_j(J) = v_i(J)$.
\item For any $i,j \in \mathcal{D}_J$ : $w_i(J) = v_j(J)$ and $w_j(J) = v_i(J)$
\end{itemize}

We didn't give a full characterization of all the proxy-phantom grading functions; We could have studied a richer proxy class (that may include the uniform median \cite{ICML2016,Freeman} excluded by our proxy) but we didn't in the paper because their characterizations (we have) are much more involved; We didn't discuss all the axioms we want for a ranking functions (such as IIA, partial SP, consistency in ranking, etc). Finally, we could have tried to characterize the methods where voters weights differ across voters and/or candidates, etc.  

We have taken as given the voting rights. It is important to understand the optimal or approximately optimal way to determine those rights, given a method. When we have a large set of equally qualified voters who just don't have the time to grade all the candidates, proposing at random a small number of candidates to each seems to be the right solution. \cite{BS2020} computed the optimal way to approximate Borda and minmax rules; LaPrimaire.org, implemented an apparently good approximation where each voter were asked to grade 5 candidates out of 12 selected at random. In EC 2023, because not all the referees are competent for all submitted papers, the rights are distributed after taking into account the declared conflicts of interests and the intensity of the preferences (expressed in the scale $[-20,20]$). 


We worked on a static environment and we were looking for fixed grading function. But many practical instances are dynamic and contains many blank votes (grading movies, hotels, restaurants, products at the internet, etc). In such instances, we may want the proxies to be dynamically updated using some ML algorithm that learns or predicts from the grades the voter already gave on some candidates and other private or social characteristic, how they will behave in other instances and use those predictions as proxies that will replace in the instances where they didn't vote. 

Delegation can allow to have a better aggregation when a small number of voters must evaluate a large number of candidate. It can easily be incorporated to our methods. If $i$ is eligible and delegates to $j$ its vote on $J$, we should respect that and let the vote of $i$ for $J$ be that of $j$ for $J$. 

\textbf{To conclude}: We built a complex but realistic and practical model where voters have different voting rights, can cast blank votes or abstain. We extended in several directions the results in \cite{M1980} and \cite{BL2011}. This leads us to a new class of methods to grade and rank, but also rises many interesting research questions that one paper is not enough to explore. 


\bibliographystyle{ACM-Reference-Format}
\bibliography{EC23.bib}
\newpage
\appendix

\section{Lemma to the proofs}

\begin{lemma}[A useful lemma]\label{useful lemma 1}
    For any SP $\varphi$ function, we can choose our $\omega_{J,S}^T$ functions such that:
    \[\forall \textbf{v}, \omega_{J,S}^T(\textbf{v}_{-T}) < \inf\mathcal{A} \Rightarrow \omega_{J,S}^T(\textbf{v}_{-T}) =\omega_{J,T}^T(\textbf{v}_{-T})\]
    and
    \[\forall \textbf{v}, \omega_{J,S}^T(\textbf{v}_{-T}) > \sup\mathcal{A} \Rightarrow \omega_{J,S}^T(\textbf{v}_{-T}) =\omega_{J,\emptyset}^T(\textbf{v}_{-T})\]
\end{lemma}

\begin{proof}
For any voting profile $\textbf{v}$, let $T = \mathcal{D''}_J$
\begin{itemize}
    \item Suppose that $\omega_{J,T}^T(\textbf{v}_{-T}) < \inf\mathcal{A}$. For all $S \subseteq T$ since $\min \{v_i(J) : i \in S\} \cup \{\omega_{J,S}^T(\textbf{v}_{-T}) \} \leq \{v_i(J) : i \in T\} \cup \{\omega_{J,T}^T(\textbf{v}_{-T}) \}$ we have that $\varphi(\textbf{v})(J)= \omega_{J,T}^T(\textbf{v}_{-T})$. Therefore, without loss of generality, we can assume that for all $S$:    
    $\omega_{J,S}^T(\textbf{v}_{-T}) =\omega_{J,T}^T(\textbf{v}_{-T})$.
    \item  Suppose that $\omega_{J,T}^T(\textbf{v}_{-T}) \geq \inf\mathcal{A}$ then for $S$ such that $\omega_{J,S}^T(\textbf{v}_{-T}) < \inf\mathcal{A}$ we have that $\min \{v_i(J) : i \in S\} \cup \{\omega_{J,S}^T(\textbf{v}_{-T}) \} < \min \{v_i(J) : i \in S\} \cup \{\inf\mathcal{A} \} \leq \{v_i(J) : i \in T\} \cup \{\omega_{J,T}^T(\textbf{v}_{-T}) \}$. Therefore, without loss of generality, we can assume that if $\omega_{J,S}^T(\textbf{v}_{-T}) \leq \inf\mathcal{A}$ then $\omega_{J,S}^T(\textbf{v}_{-T}) = \inf\mathcal{A}$.
    \item Suppose that $\omega_{J,\emptyset}^T(\textbf{v}_{-T}) > \sup\mathcal{A}$. Then for all $\emptyset \subset S \subseteq T$, $\min \{v_i(J) : i \in S\} \cup \{\omega_{J,S}^T(\textbf{v}_{-T}) \} \leq \omega_{J,\emptyset}^T(\textbf{v}_{-T})$. Therefore $\varphi(\textbf{v})(J)= \omega_{J,\emptyset}^T(\textbf{v}_{-T})$. Therefore, without loss of generality, we can assume that for all $S$:    
    $\omega_{J,S}^T(\textbf{v}_{-T}) =\omega_{J,\emptyset}^T(\textbf{v}_{-T})$.
    \item Suppose that $\omega_{J,\emptyset}^T(\textbf{v}_{-T}) \leq \sup\mathcal{A}$ then for $S$ such that $\omega_{J,S}^T(\textbf{v}_{-T}) > \sup\mathcal{A}$ we have that $\min \{v_i(J) : i \in S\} \cup \{\omega_{J,S}^T(\textbf{v}_{-T}) \} = \min \{v_i(J) : i \in S\} \cup \{\sup\mathcal{A} \}$. Therefore, without loss of generality, we can assume that if $\omega_{J,S}^T(\textbf{v}_{-T}) \geq \sup\mathcal{A}$ then $\omega_{J,S}^T(\textbf{v}_{-T}) = \sup\mathcal{A}$.
\end{itemize}

\end{proof}

\begin{lemma}\label{useful lemma 2}
    For all (SP) $\varphi$ functions, for all $S \subseteq T$ and $J$, $\forall \textbf{v}_{-T}$, we can complete to obtain $\textbf{v}$ such that $\varphi(\textbf{v})=\omega_{J,S}^T(\textbf{v}_{-T})$.
    
    Furthermore if $\omega_{J,S}^T(\textbf{v}_{-T}) \not \in \{\inf\mathcal{A},\sup\mathcal{A}\}$ then we can complete $\textbf{v}$ in such a way as to never have $i$ such that $v_i(J) = \omega_{J,S}^T(\textbf{v}_{-T})$.
\end{lemma}

\begin{proof}
Let us place ourselves in the simplified choice of $\omega_{J,S}^T$ suggested by lemma \ref{useful lemma 1}.
\begin{itemize}
    \item If $\omega_{J,T}^T(\textbf{v}_{-T}) < (\textbf{v}_{-T})(\textbf{v}_{-T})$, then no matter how we complete $\textbf{v}$ we have $\varphi(\textbf{v})= \omega_{J,T}^T(\textbf{v}_{-T})=\omega_{J,S}^T(\textbf{v}_{-T})$.
    \item If $\omega_{J,\emptyset}^T(\textbf{v}_{-T}) > (\textbf{v}_{-T})(\textbf{v}_{-T})$, then no matter how we complete $\textbf{v}$ we have $\varphi(\textbf{v})= \omega_{J,\emptyset}^T(\textbf{v}_{-T})=\omega_{J,S}^T(\textbf{v}_{-T})$.
    \item Else we have $\inf\mathcal{A}\leq \omega_{J,S}^T(\textbf{v}_{-T}) \leq\sup\mathcal{A}$. We can therefore select 
    $\alpha \leq \omega_{J,S}^T(\textbf{v}_{-T})$ and $\beta \geq \omega_{J,S}^T(\textbf{v}_{-T})$. If $i\in S$ then $v_i(J)=\beta$ and else $v_i(J)=\alpha$. We get $\varphi(\textbf{v})(J) =\omega_{S,J}^T(\textbf{v}_{-T})$. If $\omega_{S,J}^T(\textbf{v}_{-T}) \not \in \{\inf\mathcal{A},\sup\mathcal{A}\}$ then we can select $\alpha$ and $\beta$ such that $\alpha < \omega_{J,S}^T(\textbf{v}_{-T}) < \beta$.
\end{itemize}
\end{proof}

\section{Section \ref{section SP} }











\subsection{Proof for Proposition \ref{SP prop}}\label{SP prop proof}

\begin{proof}
    For $i \in \mathcal{N}$ and $J \in \mathcal{M}$. For a voting profile $\textbf{v}$ such that $v_i(J) \in \mathcal{A}$. Let $\textbf{w}=\textbf{v}[v_i = w_i : w_i(J) \in \mathcal{A}]$. If $v_i(J) < \psi(\textbf{v})(J)$ we get:
    \begin{align*}
        \psi(\textbf{v})(J) & = \mu_{g_J(\#\mathcal{D''}_J \cup \mathcal{F}_J(\textbf{v}))}(\textbf{v}(J) \cup \mathcal{F}_J(\textbf{v})) \\
        & \leq \mu_{g_J(\#\mathcal{D''}_J \cup \mathcal{F}_J(\textbf{v}))}(\textbf{w}(J) \cup \mathcal{F}_J(\textbf{v})) \\
        & = \mu_{g_J(\#\mathcal{D''}_J \cup \mathcal{F}_J(\textbf{w}))}(\textbf{w}(J) \cup \mathcal{F}_J(\textbf{w})) \\
        & = \psi(\textbf{w})(J)
    \end{align*}

The $\psi(\textbf{v})(J) < v_i(J)$ is symmetric. Therefore we verify SP.
    
\end{proof}

\subsection{Proof for proposition \ref{SP charac prop}}\label{SP charc prop proof}

\begin{proof}
Let $\psi$ be a phantom-proxy mechanism defined by:
\[\forall \textbf{v},\psi(\textbf{v})(J)=\mu_{g_J(\#\mathcal{D''}_J \cup \mathcal{F}_J(\textbf{v}))}(\textbf{v}(J) \cup \mathcal{F}_J(\textbf{v})).\]

For any $\textbf{v}$, let $n=\#(\mathcal{D''}_J \cup \mathcal{F}_J(\textbf{v}))$ and $p=g(n)$ and $\alpha=\psi(\textbf{v})$. We have that $\alpha$ is the $p$-th smallest element of $\textbf{v}(J) \cup \mathcal{F}_J(\textbf{v})$. 

Let $\varphi$ be the SP grading function defined by the phantom-mapping given in the characterization.

\begin{align*}
    \varphi(\textbf{v})(J) = & \max_{S \subseteq T} \min \{v_i(J) : i \in S \} \cup \{\omega_{J,S}^T(\textbf{v}_{-T}) \} \\
    = & \max_{S \subseteq T} \min \{v_i(J) : i \in S \} \cup \{\mu_{\#S - \#T +p}(\mathcal{F}_J(\textbf{v}))\} 
\end{align*}

Suppose that $\alpha \not \in \textbf{v}(J)$ then let $m$ be the number of elements of $\mathcal{F}_J(\textbf{v})$ counted as less than $\alpha$ when determining $\psi(\textbf{v})(J)$. There are therefore $p-m$ elements of $\textbf{v}(J)$ counted as less than $\alpha$. For $S$ the set that does not contain the $p-m$ voters associated to the smallest elements of the multi-set $\textbf{v}(J)$. $\#S = \#T - p + m$, therefore $ \mu_{\#S - \#T +p}(\mathcal{F}_J(\textbf{v})) = \alpha$. It follows that:
\[\min \{v_i(J) : i \in S \} \cup \{\mu_{\#S - \#T +p}(\mathcal{F}_J(\textbf{v}))\} = \alpha\]
And for any other $S' \subseteq T$ we have :

\[\min \{v_i(J) : i \in S' \} \cup \{\mu_{\#S' - \#T +p}(\mathcal{F}_J(\textbf{v}))\} \leq \alpha.\]

Therefore we have that $\varphi(\textbf{v})(J) = \alpha.$

If $\alpha \in \textbf{v}(J)$ then let $m$ be the number of elements of $\mathcal{F}_J(\textbf{v})$ counted as less than $\alpha$ when determining $\psi(\textbf{v})(J)$. There are therefore $p-m-1$ elements of $\textbf{v}(J)$ counted as less than $\alpha$. For $S$ the set that does not contain the $p-m-1$ voters associated to the smallest elements of the multi-set $\textbf{v}(J)$. $\#S = \#T - p + m +1$, therefore $ \mu_{\#S - \#T +p}(\mathcal{F}_J(\textbf{v})) \geq \alpha$. It follows that:
\[\min \{v_i(J) : i \in S \} \cup \{\mu_{\#S - \#T +p}(\mathcal{F}_J(\textbf{v}))\} = \alpha\]
And for any other $S' \subseteq T$ we have :

\[\min \{v_i(J) : i \in S' \} \cup \{\mu_{\#S' - \#T +p}(\mathcal{F}_J(\textbf{v}))\} \leq \alpha.\]

Therefore we have that $\varphi(\textbf{v})(J) = \alpha.$

\end{proof}

\section{Section \ref{section non graders} }
\subsection{Proof for BV theorem \ref{BV theo}}\label{BV theo proof}
\begin{proof}
    $\Rightarrow:$ By using our lemma \ref{useful lemma 2} we can show that for all $\textbf{v}$ with $v_i(J) \in \emptyset$ and $\textbf{w}=\textbf{v}[v_i(I) =\otimes]$. For all $S, T, J$ we have that $\omega_{J,S}^T(\textbf{v})=\omega_{J,S}^T(\textbf{w})$. Therefore phantom-mapping verify (BV). 
    $\Leftarrow:$ Immediate
\end{proof}

\subsection{Proof for BV property \ref{BV prop}}\label{BV prop proof}
Immediate due to characterization of the phantom-mappings \ref{SP charc prop proof}.

\subsection{Proof for SI theorem \ref{SI theo}}\label{SI theo proof}
\begin{proof}
    $\Rightarrow:$ By using our lemma \ref{useful lemma 2} we can show that for all $\textbf{v}$ with $v_i(J) = \circ$ and $\textbf{w}=\textbf{v}[\forall J, v_j(J) =\emptyset]$. For all $S, T, J$ we have that $\omega_{J,S}^T(\textbf{v})=\omega_{J,S}^T(\textbf{w})$. Therefore phantom-mapping verify (SI). 
    $\Leftarrow:$ Immediate.
\end{proof}

\subsection{Proof for SI property \ref{SI prop}}\label{SI prop proof}
Immediate due to characterization of the phantom-mappings \ref{SP charc prop proof} and the definition of the phantom-proxy mechanism.

\subsection{Proof for SC theoretical lemma \ref{SC theo lemma}}\label{SC theo lemma proof}

\begin{proof}
$\Rightarrow:$
For any $J$,$S$ and $T$ and $i \not \in T$, let $\textbf{v}$ be such that $v_i(J)=\circ$, for $j \in S$, $v_j(J) = \alpha$ such that $\omega_{J,\emptyset}^T(\textbf{v}_{-T}) \leq \alpha \leq \omega_{J,S}^T(\textbf{v}_{-T})$ and for $j \in T-S$, $v_j(J) = \beta \leq  \alpha$. Suppose that $\beta < \alpha$, we have that $\varphi(\textbf{v})(J) = \alpha$. Let $\textbf{w} = \textbf{v}[v_i(J):=\alpha]$, by SC we get $\varphi(\textbf{w})(J) = \alpha$. As such we have
$\omega_{J,\emptyset}^{T \cup \{i\}}(\textbf{w}_{-T \cup \{i\}}) \leq \alpha \leq \omega_{J,S \cup \{i\}}^{T \cup \{i\}}(\textbf{w}_{-T \cup \{i\}})$.
If we could not choose $\beta<\alpha$ then $\omega_{J,\emptyset}^T(\textbf{v}_{-T}) = \alpha =\inf\mathcal{A}$ therefore by SC we have $\omega_{J,\emptyset}^{T\cup\{i\}}(\textbf{v}_{-T \cup \{i\}}) \leq \alpha \leq \omega_{J,S}^{T\cup\{i\}}(\textbf{v}_{-T \cup \{i\}})$.

For any $J$,$S$ and $T$ and $i \not \in T$, let $\textbf{v}$ be such that $v_i(J)=\circ$, for $j \in S$, $v_j(J) = \alpha$ and for $j \in T-S$, $v_j(J) = \beta$ such that $\omega_{J,S}^T(\textbf{v}_{-T}) \leq \beta \leq \omega_{J,T}^{T}(\textbf{v}_{-T})$ and $\beta \leq \alpha$. Suppose that $\beta < \alpha$. We have that $\varphi(\textbf{v})(J) = \beta$. Let $\textbf{w} = \textbf{v}[v_i(J):=\beta]$, by SC we get $\varphi(\textbf{w})(J) = \beta$. As such we have
$\omega_{J,S}^{T \cup \{i\}}(\textbf{w}_{-(T \cup \{i\})}) \leq \beta \leq \omega_{J,T \cup \{i\}}^{T \cup \{i\}}(\textbf{v}_{-(T \cup \{i\})})$.
If we could not choose $\beta < \alpha$ then $\omega_{J,T}^T(\textbf{v}_{-T}) = \beta = \sup\mathcal{A}$ therefore by SC we have  $\omega_{J,S}^{T\cup\{i\}}(\textbf{v}_{-(T\cup\{i\})}) \leq \alpha \leq \omega_{J,T}^{T\cup\{i\}}(\textbf{v}_{-(T\cup\{i\})})$.

$\Leftarrow:$ Suppose we have our inequalities. Let $\textbf{v}$ be such that $v_i(J) = \circ$ and $\varphi(\textbf{v})(J) = \alpha \in \mathcal{A}$. Let $\textbf{w}=\textbf{v}[v_i(J):=\alpha]$. Then we have that:
\begin{align*}
    \varphi(\textbf{v})(J) = & \max_{S \in T} \min \{v_i(J) : i \in S \} \cup \{\omega_{J,S}^{T}(\textbf{v}_{-T}) \} \\
    = & \min \{v_j(J) : j \in S_v \} \cup \{\omega_{J,S_v}^{T}(\textbf{v}_{-T}) \} \\
    = & \min \{v_j(J) : j \in S_v \} \cup \{\alpha\} \cup \{\omega_{J,S_v}^{T}(\textbf{v}_{-T}) \} \\
    \leq & \min \{w_j(J) : j \in S_v \cup \{i\} \} \cup \{\omega_{J,S_v \cup \{i\}}^{T\cup\{i\}}(\textbf{v}_{-(T\cup\{i\})}) \} \\
    \leq & \varphi(\textbf{w})(J)
\end{align*}

\begin{align*}
    \varphi(\textbf{w})(J) = & \max_{S \in T \cup \{i\}} \min \{w_i(J) : i \in S \} \cup \{\omega_{J,S}^{T\cup\{i\}}(\textbf{w}_{-(T\cup\{i\})}) \} \\
    = & \min \{w_j(J) : j \in S_w \} \cup \{\omega_{J,S_w}^{T\cup\{i\}}(\textbf{w}_{-(T\cup\{i\})}) \} \\
    = & \min \{v_j(J) : j \in S_w - \{i\}\} \cup \{\alpha\} \cup \{\omega_{J,S_w}^{T\cup\{i\}}(\textbf{w}_{-(T\cup\{i\})}) \} \\
    \leq & \min \{v_j(J) : j \in S_w - \{i\}\} \cup \{\omega_{J,S_w}^{T}(\textbf{v}_{-T}) \} \\
    \leq  & \varphi(\textbf{v})(J)
\end{align*}

Therefore $\varphi(\textbf{w})(J)=\alpha$, we verify SC.

\end{proof}

\subsection{Proof for SC theorem \ref{SC theo}}\label{SC theo proof}

 \begin{proof}
 This is an immediate consequence of the characterization \ref{SC theo lemma}.
 \end{proof}

 \subsection{Proof for SC proposition \ref{SC prop}}\label{SC prop proof}

  \begin{proof}
$\Rightarrow:$
Let $i$ be the absentee voter. When player $i$ changed to a vote $\alpha$ for $J$ we have 2 options, either there was a proxy for him ($f_{i,J}(v_i)\in \mathcal{B}$) in which case SP provides us that we are SC or $\mathcal{F}_J(\textbf{v}) =\mathcal{F}_J(\textbf{v}[v_i(J):=\alpha])$. 

Let us recall that the phantom-mappings are characterized as for any $S$ and $T$ and $\textbf{v}$:
$k=\#S - \#T +g_J(\#\mathcal{D''}_J + \#\mathcal{F}_J(\textbf{v}))$.

\begin{itemize}
    \item If $k \leq 0$ then $\omega_{J,S}^T(\textbf{v}_{-T}) = \inf\mathcal{B}$.
    \item If $k > \#\mathcal{F}_J(\textbf{v})$ then $ \omega_{J,S}^T(\textbf{v}_{-T})  = \sup\mathcal{B}$
    \item Else $\omega_{J,S}^T(\textbf{v}_{-T}) = \mu_{k}(\mathcal{F}_J(\textbf{v}))$
\end{itemize}

Since the set of proxy phantoms has not changed and that according to the SC characterization we need to ensure a series of inequality we need. Let $k'=\#S - \#T +g_J(\#\mathcal{D''}_J + \#\mathcal{F}_J(\textbf{v})+1)$.

\begin{itemize}
    \item If $k' \leq 0$ then $\omega_{J,S}^{T\cup\{i\}}(\textbf{v}_{-(T\cup\{i\})}) = \inf\mathcal{B}$. It follows that we need $\omega_{J,S}^{T}(\textbf{v}_{-T}) \leq \inf\mathcal{A}$. Since $g_J$ cannot depend on the values of $\mathcal{F}_J$, we therefore need $k \leq 0$. As such for all $p \in \mathbf{N}$, $g_J(p) \leq g_J(p+1)$.
    \item If $k' \geq \#\mathcal{F}_J(\textbf{v})$ then $ \omega_{J,S\cup \{i\}}^{T \cup \{i\}}(\textbf{v}_{-(T \cup \{i\})})  = \sup\mathcal{B}$. It follows that we need $\omega_{J,S}^{T}(\textbf{v}_{-T}) > \sup\mathcal{A}$. Since $g_J$ cannot depend on the values of $\mathcal{F}_J$, we therefore need $k \geq \#\mathcal{F}_J(\textbf{v})$. As such for all $p \in \mathbf{N}$, $g_J(p) +1 \geq g_J(p+1)$.
\end{itemize}

We have shown that for all $p$ we have $g_J(p+1) -1 \leq g_J(p) \leq g_J(p+1)$. As such $g_J(p+1) \in \{g_J(p),g_J(p)+1\}$.

$\Leftarrow:$ Suppose that for all $p\in\mathbf{N}$ we have $g_J(p+1) \in \{g_J(p),g_J(p)+1\}$. Then for all $S$ we have $\omega_{J,S}^{T\cup\{i\}}(\textbf{v}_{-T\cup\{i\}}) \leq \omega_{J,S}^{T}(\textbf{v}_{-T}) \leq \omega_{J,S\cup\{i\}}^{T\cup\{i\}}(\textbf{v}_{-T\cup\{i\}})$. The characterization of SC gives the rest.

 \end{proof}
 
\subsection{Proof for P remark \ref{P remark}}\label{P remark proof}

\begin{proof}
 Suppose that $\varphi$ verifies the participation property (P).
For any $\textbf{w}$ such that $w_i(J) = \circ$. Let $\alpha = \varphi(\textbf{w})(J) \in \mathcal{A}$ and let $\textbf{v}=\textbf{w}[v_i(J) :=\alpha]$.

If we suppose that $\varphi(\textbf{v})(J) > v_i(J)=\alpha$. Then by (P) we have that $\alpha = \varphi(\textbf{w})(J) \geq \varphi(\textbf{v})(J) >\alpha$. This is absurd. A symmetrical proof shows that $\varphi(\textbf{v})(J) < v_i(J)=\alpha$ is also absurd. As such we must have $\varphi(\textbf{v})(J) =\alpha$.

We have therefore that $\varphi$ verifies the "silent consent" rule.
\end{proof}

\subsection{Proof for P theorem \ref{P theo}}\label{P theo proof}

\begin{proof}
$\Rightarrow:$ Let us consider a voting profile $\textbf{v}$ such that $v_i(J) = \alpha \in \mathcal{A}$. Let us assume that $\varphi(\textbf{v})(J) = \omega_{J,S}^{\mathcal{D''}_J}(\textbf{v}_{-\mathcal{D''}_J})$ and that $\forall i, \varphi(\textbf{v})(J) \neq v_i(J)$ (We can always obtain this result by moving the $v_i(I)$ with $I \neq J$ around unless $\omega_{J,S}^{\mathcal{D''}_J}(\textbf{v}_{-\mathcal{D''}_J}) \in\{\inf \mathcal{A},\sup \mathcal{A}\}$). Let the profile $\textbf{w}=\textbf{v}[w_i=\circ]$.

First case:
 $\omega_{J,S}^{\mathcal{D''}_J}(\textbf{v}_{-\mathcal{D''}_J}) < \alpha$. We have $i \in S$. Suppose that 
 \[\omega_{J,S}^{\mathcal{D''}_J(\textbf{v})}(\textbf{v}_{-\mathcal{D''}_J}) < \omega_{J,S-\{i\}}^{\mathcal{D''}_J(\textbf{w})}(\textbf{w}_{-\mathcal{D''}_J}).\]
 If $\omega_{J,S}^{\mathcal{D''}_J}(\textbf{v}_{-\mathcal{D''}_J}) = \sup \mathcal{A}$. This implies that $\varphi(\textbf{w})(J)=\omega_{J,\emptyset}^{\mathcal{D''}_J(\textbf{w})}(\textbf{w}_{-\mathcal{D''}_J}) > \sup \mathcal{A} = \varphi(\textbf{v})(J)$.
 Else we have that:
 \begin{align*}
     \varphi(\textbf{v})(J) = & \omega_{J,S}^{\mathcal{D''}_J}(\textbf{v}_{-\mathcal{D''}_J}) \\
     < & \min \{v_j(J) : j \in S - \{i\}\} \cup \{ \omega_{J,S-\{i\}}^{\mathcal{D''}_J(\textbf{w})}(\textbf{w}_{-\mathcal{D''}_J}) \}
     \leq \varphi(\textbf{w})(J)\\
 \end{align*}
We can conclude that :
 \[\omega_{J,S}^{\mathcal{D''}_J(\textbf{v})}(\textbf{v}_{-\mathcal{D''}_J}) \geq \omega_{J,S-\{i\}}^{\mathcal{D''}_J(\textbf{w})}(\textbf{w}_{-\mathcal{D''}_J}).\]

Second case:
 $\alpha < \omega_{J,S}^{\mathcal{D''}_J}(\textbf{v}_{-\mathcal{D''}_J})$. We have $i \not \in S$. Suppose that 
 \[\omega_{J,S}^{\mathcal{D''}_J(\textbf{v})}(\textbf{v}_{-\mathcal{D''}_J}) > \omega_{J,S}^{\mathcal{D''}_J(\textbf{w})}(\textbf{w}_{-\mathcal{D''}_J}).\]
 If $\omega_{J,S}^{\mathcal{D''}_J}(\textbf{v}_{-\mathcal{D''}_J}) = \inf \mathcal{A}$. This implies that $\varphi(\textbf{w})(J)=\omega_{J,{\mathcal{D''}_J(\textbf{w})}}^{\mathcal{D''}_J(\textbf{w})}(\textbf{w}_{-\mathcal{D''}_J}) < \inf \mathcal{A} = \varphi(\textbf{v})(J)$.
 Else we have that:
 \begin{align*}
     \varphi(\textbf{w})(J) \leq & \min \{w_j(J) : j \in S \} \cup \{ \omega_{J,S}^{\mathcal{D''}_J(\textbf{w})}(\textbf{w}_{-\mathcal{D''}_J}) \} \\
     \leq & \omega_{J,S}^{\mathcal{D''}_J(\textbf{w})}(\textbf{w}_{-\mathcal{D''}_J}) \\
     < & \omega_{J,S}^{\mathcal{D''}_J(\textbf{v})}(\textbf{v}_{-\mathcal{D''}_J}) \\
     = & \varphi(\textbf{v})(J)\\
 \end{align*}
We can conclude that :
 \[\omega_{J,S}^{\mathcal{D''}_J(\textbf{v})}(\textbf{v}_{-\mathcal{D''}_J}) \leq \omega_{J,S}^{\mathcal{D''}_J(\textbf{w})}(\textbf{w}_{-\mathcal{D''}_J}).\]

$\Leftarrow:$ Suppose that we have desired characterization. For any $J \in \mathcal{M}$, let $\textbf{v}$ be a voting profile with $v_i(J) = \alpha$. Let $\textbf{w}$ obtained by $\textbf{v}$ by $w_i(J)= \circ$ everything else equal.

Suppose that $\varphi(\textbf{v})(J) > v_i(J)$. It follows that for $S$ such that :
\[\varphi(\textbf{v})(J)= \max_{S \subseteq \mathcal{D''}_J(\textbf{v})} \min \{v_j(J) : j \in S\}\ \cup \{\omega_{J,S}^{\mathcal{D''}_J(\textbf{v})}(\textbf{v}_{-\mathcal{D''}_J})\}\]
we have that $i \not \in S$. As such $S \subseteq \mathcal{D''}_J(\textbf{w})$. As such:

\begin{align*}
    \varphi(\textbf{w})(J)= &\max_{S \subseteq \mathcal{D''}_J(\textbf{w})} \min \{w_j(J) : j \in S\} \cup \{\omega_{J,S}^{\mathcal{D''}_J(\textbf{w})}(\textbf{w}_{-\mathcal{D''}_J})\} \\
    \geq & \max_{S \subseteq \mathcal{D''}_J(\textbf{w})} \min \{v_j(J) : j \in S\}\cup \{\omega_{J,S}^{\mathcal{D''}_J(\textbf{v})}(\textbf{v}_{-\mathcal{D''}_J})\} \\
    \geq & \max_{S \subseteq \mathcal{D''}_J(\textbf{v})} \min \{v_j(J) : j \in S\} \cup \{\omega_{J,S}^{\mathcal{D''}_J(\textbf{v})}(\textbf{v}_{-\mathcal{D''}_J})\} \\
    = & \varphi(\textbf{v})(J)
\end{align*}

Conversely if $\varphi(\textbf{v})(J) < v_i(J)$. Due to the monotonicity of the $\omega_{J,S}^T$ functions in respect to the sets $S$ we know that we can find $S$ such that $\varphi(\textbf{v})(J)= \max_{S \subseteq \mathcal{D''}_J(\textbf{v})} \min \{v_j(J) : j \in S\} \cup \{\omega_{J,S}^{\mathcal{D''}_J(\textbf{v})}(\textbf{v}_{-\mathcal{D''}_J})\}$ and $i \in S$. As such:

\begin{align*}
    \varphi(\textbf{w})(J)= &\max_{S \subseteq \mathcal{D''}_J(\textbf{w})} \min \{w_j(J) : j \in S\} \cup \{\omega_{J,S}^{\mathcal{D''}_J(\textbf{w})}(\textbf{w}_{-\mathcal{D''}_J})\} \\
    \leq & \max_{\{i\} \subseteq S \subseteq \mathcal{D''}_J(\textbf{v})} \min \{v_j(J) : j \in S\} \cup \{\omega_{J,S}^{\mathcal{D''}_J(\textbf{v})}(\textbf{v}_{-\mathcal{D''}_J})\} \\
    \leq & \max_{S \subseteq \mathcal{D''}_J(\textbf{v})} \min \{v_j(J) : j \in S\} \cup \{\omega_{J,S}^{\mathcal{D''}_J(\textbf{v})}(\textbf{v}_{-\mathcal{D''}_J})\} \\
    = & \varphi(\textbf{v})
\end{align*}

So the participation property is verified.

 \end{proof}

\subsection{Proof for P and FP propositions \ref{P prop}}\label{FP prop proof}

\begin{proof}
The proof is the same as for (SC) see \ref{SC prop proof}
\end{proof}

\subsection{Proof for JD lemma \ref{JD lemma}}\label{JD lemma proof}

\begin{proof}
    (JD) provides that the outcome of a $\omega_J$ function only depends on the votes that were associated to $J$. Therefore in order to determine the value of $\omega_{J,S}^T$ we are only interested in the votes $v_i(J) \in \{\emptyset,\otimes,\circ\}$.
\end{proof}

\subsection{Proof for JD prop \ref{JD prop}}\label{JD prop proof}
\begin{proof}
    Immediate due to the characterization \ref{SP charac prop}.
\end{proof}
\section{Section \ref{section prop}}

\subsection{Proof for U theorem \ref{U theo}}\label{U theo proof}

\begin{proof}
        Suppose that we are unanimous and that there exists a profile $\textbf{v}$ such that $\omega_{J,\emptyset}^T(\textbf{v}_{-\mathcal{D''}_J}) > \inf \mathcal{A}$. Let $\alpha \in \mathcal{A}$ be such that $\alpha < \omega_{J,\emptyset}^T(\textbf{v}_{-\mathcal{D''}_J})$. Let $\textbf{w}=\textbf{v}[v_i(J) : i \in \mathcal{D''}_J]$. Then by (U) we must have $\alpha=\varphi(\textbf{w})(J)=\omega_{J,\emptyset}^T(\textbf{v}_{-\mathcal{D''}_J})$. This is absurd, therefore $\omega_{J,\emptyset}^T(\textbf{v}_{-\mathcal{D''}_J}) \leq \inf \mathcal{A}$. 
        
        The proof for $\inf \mathcal{A} \wedge \omega_{J,T}^T\geq\sup\mathcal{A}$ is done symmetrically.
\end{proof}

\subsection{Proof for N theorem \ref{N theo}}\label{N theo proof}

\begin{proof}
$\Rightarrow:$
    Let $\varphi:\mathcal{O}^{\mathcal{M} \times \mathcal{N}} \rightarrow {(\mathcal{B} \cup \{\emptyset\})}^\mathcal{M}$ be (SP,N). For any $\textbf{v}$, according to neutrality if $\mathcal{D}_I=\mathcal{D}_J$ then we can can switch the votes of $I$ and $J$ without changing the outcome. Let $\textbf{w}$ be obtained from $\textbf{v}$ by the swapping of votes for $I$ with those for $J$.
    Let us first consider $K$ that is not $I$ or $J$. By using the lemma \ref{useful lemma 2}, we get $\omega_{K,S}^{\mathcal{D''}_K}(\textbf{v}_{-\mathcal{D''}_K})=\omega_{K,S}^{\mathcal{D''}_K}(\textbf{w}_{-\mathcal{D''}_K})$. So when we swapped $J$ and $I$ the outcome of $\omega_{K,S}^{\mathcal{D''}_K}$ was not affected. Therefore $\omega_{K,S}^{\mathcal{D''}_K}$ is neutral.

    Let us now consider $I$ and $J$, by neutrality we have that $\varphi(\textbf{v})(J)=\varphi(\textbf{w})(I)$. By using lemma \ref{useful lemma 2} we can get:  $\omega_{J,S}^{\mathcal{D''}_J}(\textbf{v}_{-\mathcal{D''}_J})=\omega_{I,S}^{\mathcal{D''}_I}(\textbf{w}_{-\mathcal{D''}_I})$. 
    Since all phantom-mappings $\omega_{K,S}^T$ functions are neutral, we therefore have that $\omega_{J,S}^{\mathcal{D''}_J}(\textbf{v}_{-\mathcal{D''}_J})=\omega_{I,S}^{\mathcal{D''}_I}(\textbf{w}_{-\mathcal{D''}_I})=\omega_{I,S}^{\mathcal{D''}_I}(\textbf{v}_{-\mathcal{D''}_J})$.
    It follows that $\omega_{J,S}^{\mathcal{D''}_J}=\omega_{I,S}^{\mathcal{D''}_I}$. This proves the existence of the $\omega_{U,S}^{T}$ functions.

$\Leftarrow:$  
For $I$ and $J$ such that $\mathcal{D}_I=\mathcal{D}_J$, for a given voting profile $\textbf{v}$, let $\textbf{w}=\textbf{v}[\forall i, v_i(I):=v_i(J),v_i(J):=v_i(I)]$.
    \begin{align*}
        \varphi(\textbf{v})(J)=& \max_{S\subseteq \mathcal{D''}_J(\textbf{v})} \min \{v_i(J) : i \in S\} \cup \{\omega_{J,S}^{\mathcal{D''}_J}(\textbf{v}_{-\mathcal{D''}_J})\} \\
        =&\max_{S\subseteq {\mathcal{D''}_I}(\textbf{w})}\min \{w_i(I) : i \in S\} \cup \{\omega_{I,S}^{\mathcal{D''}_I}(\textbf{w}_{\mathcal{D''}_I})\}\\
        =&\varphi(\textbf{w})(J).
    \end{align*}
\end{proof}

\subsection{Proof for N proposition \ref{N prop}}\label{N prop proof}

\begin{proof}
$\Rightarrow:$ Let $\psi$ be an (SP,N) phantom-proxy mechanism. For any $\textbf{v}$. Let $U =\mathcal{D}_J=\mathcal{D}_I$, let us switch $I$ and $J$ to obtain $\textbf{w}=\textbf{v}[J:=v_i(I), I:=v_i(J) : \forall i]$. For $K \not \in \{I,J\}$, $\psi(\textbf{v})(K)=\psi(\textbf{w})(K)$. We can use lemma \ref{useful lemma 2} to obtain that:

\begin{align*}
    \omega_{K,S}^T(\textbf{v}_{-T})= & \mu_{(\#S - \# T - g_K(\#T + \#\mathcal{F}_K(\textbf{v}))}(\mathcal{F}_K(\textbf{v})) \\
    & \mu_{(\#S - \# T - g_K(\#T+\#\mathcal{F}_K(\textbf{w}))}(\mathcal{F}_K(\textbf{w})) \\
    = & \omega_{K,S}^T(\textbf{w}_{-T})
\end{align*}

Therefore for all $k$, we have $\mu_k(\mathcal{F}_K(\textbf{v})) = \mu_k(\mathcal{F}_K(\textbf{w}))$. As such we have $f_{i,K}(v_i) =f_{i,K}(w_i)$. Therefore all $f_{i,K}$ functions are neutral.

\begin{align*}
    \omega_{U,S}^T(\textbf{v}_{-T})= & \mu_{(\#S - \# T - g_I(\#T + \#\mathcal{F}_I(\textbf{v}))}(\mathcal{F}_I(\textbf{v})) \\
    & \mu_{(\#S - \# T - g_I(\#T+ \#\mathcal{F}_I(\textbf{w}))}(\mathcal{F}_I(\textbf{w})) \\
    & \mu_{(\#S - \# T - g_J(\#T + \#\mathcal{F}_J(\textbf{w}))}(\mathcal{F}_J(\textbf{w})) \\
    = & \omega_{U,S}^T(\textbf{w}_{-T})
\end{align*}

Therefore for all $k$, we have $\mu_k(\mathcal{F}_I(\textbf{v})) = \mu_{(k+c)}(\mathcal{F}_J(\textbf{w}))$. By considering extreme values we find that $c=0$. As such $g_I(k)=g_J(k)$ for $k \leq \#T + \max \#\mathcal{F}_I(\textbf{v})$ and $f_{i,I}(\textbf{v})=f_{i,J}(\textbf{w}) = f_{i,J}(\textbf{v})$. As such we have $f_{i,J}=f_{i,I}$.



$\Leftarrow:$ Suppose that we have our characterization. Let $\textbf{v}$ be such that $\mathcal{D}_J=\mathcal{D}_I$ and let $\textbf{w}$ be such that $\textbf{w}$ was obtained from $\textbf{v}$ by switching $I$ and $J$. In the worst case we replace a proxy vote by the same proxy vote.

\end{proof}

\subsection{Proof for SN theorem \ref{SN theo}}\label{SN theo proof}

\begin{proof}
$\Rightarrow:$
    Let $\varphi:\mathcal{O}^{\mathcal{M} \times \mathcal{N}} \rightarrow {(\mathcal{B} \cup \{\emptyset\})}^\mathcal{M}$ be (SP,SN). For any $\textbf{v}$, according to strong neutrality we can can switch the votes of $I$ and $J$ without changing the outcome. Let $\textbf{w}$ be obtained from $\textbf{v}$ by swapping $I$'s votes with those for $J$.
    Let us first consider $K$ that is not $I$ or $J$. By using the lemma \ref{useful lemma 2}, we get $\omega_{K,S}^{\mathcal{D''}_K}(\textbf{v}_{-\mathcal{D''}_K})=\omega_{K,S}^{\mathcal{D''}_K}(\textbf{w}_{-\mathcal{D''}_K})$. So when we swapped $J$ and $I$ the outcome of $\omega_{K,S}^{\mathcal{D''}_K}$ was not affected. Therefore $\omega_{K,S}^{\mathcal{D''}_K}$ is strongly neutral.

    Let us now consider $I$ and $J$, by strong neutrality we have that $\varphi(\textbf{v})(J)=\varphi(\textbf{w})(I)$. By using lemma \ref{useful lemma 2} we can get:  $\omega_{J,S}^{\mathcal{D''}_J}(\textbf{v}_{-\mathcal{D''}_J})=\omega_{I,S}^{\mathcal{D''}_I}(\textbf{w}_{-\mathcal{D''}_I})$. 
    Since all phantom-mappings $\omega_{K,S}^T$ functions are strong neutral, we therefore have that $\omega_{J,S}^{\mathcal{D''}_J}(\textbf{v}_{-\mathcal{D''}_J})=\omega_{I,S}^{\mathcal{D''}_I}(\textbf{w}_{-\mathcal{D''}_I})=\omega_{I,S}^{\mathcal{D''}_I}(\textbf{v}_{-\mathcal{D''}_J})$.
    It follows that $\omega_{J,S}^{\mathcal{D''}_J}=\omega_{I,S}^{\mathcal{D''}_I}$. As such the phantom-mappings do not depend on $J$.

$\Leftarrow:$
For $I$ and $J$, for a given voting profile $\textbf{v}$, let $\textbf{w}=\textbf{v}[\forall i, v_i(I):=v_i(J),v_i(J):=v_i(I)]$.
    \begin{align*}
        \varphi(\textbf{v})(J)=& \max_{S\subseteq \mathcal{D''}_J(\textbf{v})} \min \{v_i(J) : i \in S\} \cup \{\omega_{J,S}^{\mathcal{D''}_J}(\textbf{v}_{-\mathcal{D''}_J})\} \\
        =&\max_{S\subseteq {\mathcal{D''}_I}(\textbf{w})}\min \{w_i(I) : i \in S\} \cup \{\omega_{I,S}^{\mathcal{D''}_I}(\textbf{w}_{\mathcal{D''}_I})\}\\
        =&\varphi(\textbf{w})(J).
    \end{align*}
\end{proof}

\subsection{Proof for SN proposition \ref{SN prop}}\label{SN prop proof}

\begin{proof}
$\Rightarrow:$ Let $\psi$ be an (SP,SN) phantom-proxy mechanism. 
For any $\textbf{v}$. Let us switch $I$ and $J$ to obtain $\textbf{w}=\textbf{v}[J:=v_i(I), I:=v_i(J) : \forall i]$. For $K \not \in \{I,J\}$ $\psi(\textbf{v})(K)=\psi(\textbf{w})(K)$.

\begin{align*}
    \omega_{S}^T(\textbf{v}_{-T})= & \mu_{(\#S - \# T - g_K(\#T + \#\mathcal{F}_K(\textbf{v}))}(\mathcal{F}_K(\textbf{v})) \\
    & \mu_{(\#S - \# T - g_K(\#T+\#\mathcal{F}_K(\textbf{w}))}(\mathcal{F}_K(\textbf{w})) \\
    = & \omega_{S}^T(\textbf{w}_{-T})
\end{align*}

Therefore for all $k$, we have $\mu_k(\mathcal{F}_K(\textbf{v})) = \mu_k(\mathcal{F}_K(\textbf{w}))$. As such we have $f_{i,K}(v_i) =f_{i,K}(w_i)$. Therefore all $f_{i,K}$ functions are strong neutral.

\begin{align*}
    \omega_{S}^T(\textbf{v}_{-T})= & \mu_{(\#S - \# T - g_I(\#T + \#\mathcal{F}_I(\textbf{v}))}(\mathcal{F}_I(\textbf{v})) \\
    & \mu_{(\#S - \# T - g_I(\#T+ \#\mathcal{F}_I(\textbf{w}))}(\mathcal{F}_I(\textbf{w})) \\
    & \mu_{(\#S - \# T - g_J(\#T + \#\mathcal{F}_J(\textbf{w}))}(\mathcal{F}_J(\textbf{w})) \\
    = & \omega_{S}^T(\textbf{w}_{-T})
\end{align*}

This remains true no matter the size of $\mathcal{F}_J(\textbf{v})$, as such $g_I=g_J$.

$\Leftarrow:$ Suppose that we have our characterization. Let $\textbf{v}$ and $\textbf{w}$ such that $\textbf{w}$ was obtained from $\textbf{v}$ by switching  $I$ and $J$.

\begin{align*}
    \psi(\textbf{v})(I) = & \mu_{g_J(\#\mathcal{D''}_J(\textbf{v}) \cup \mathcal{F}_J(\textbf{v}))}(\textbf{v}(J) \cup \mathcal{F}_J(\textbf{v})) \\
    = & \mu_{g(\#\mathcal{D''}_J(\textbf{v}) \cup \mathcal{F}_J(\textbf{v}))}(\textbf{v}(J) \cup \mathcal{F}_J(\textbf{v})) \\
    = & \mu_{g(\#\mathcal{D''}_J(\textbf{w}) \cup \mathcal{F}_J(\textbf{w}))}(\textbf{w}(J) \cup \mathcal{F}_J(\textbf{w})) \\
    =\psi(\textbf{v})(I)
\end{align*}
\end{proof}

\subsection{Proof for A theorem \ref{A theo}}\label{A theo proof}

\begin{proof}
$\Rightarrow:$ 

For any given $J$ we have the following:
\begin{itemize}
    \item If both votes are in $\mathcal{A}$ we have that there is no change to the value of the $\omega_{J,S}$ functions.
    \item If both votes are in $\circ,\otimes$ (or both in $\emptyset$ then we can see that the $\omega$ functions are A. this extends without loss of generality.
    \item If 
\end{itemize}


Let us consider any $\textbf{v}$ such that $\mathcal{C}_i=\mathcal{C}_j$ and $\textbf{w}$ that was obtained by switching the ballots for $i$ and $j$ of $\textbf{v}$. By anonymity we have $\varphi(\textbf{v})=\varphi(\textbf{w})$.

For any $J$, let us use lemma \ref{useful lemma 2} to obtain $\varphi(\textbf{v})(J)=\omega_{J,S}^{\mathcal{D''}_J(\textbf{v})}(\textbf{v}_{-\mathcal{D''}_J})$. 

If we consider $v_i(J),v_j(J) \in \mathcal{A}$. Then we have $\mathcal{D''}_J(\textbf{v})= \mathcal{D''}_J(\textbf{w})=T$. We have that:

\[ \varphi(\textbf{w})(J) = \max_{S' \subseteq T} \min \{w_i: i \in S\} \cup \{\omega_{J,S'}^T(\textbf{w}_{-T}) \} \]

By A and the lemma \ref{useful lemma 2} we therefore have $\omega_{J,S}^T(\textbf{v}_{-T}) =\omega_{J,S'}^T(\textbf{w}_{-T}) $.

\[\max_{S' \subseteq T} \min \{w_i: i \in S'\} \cup \{\omega_{J,S'}^T(\textbf{w}_{-T}) \} \\
    = \max_{S \subseteq T} \min \{v_i: i \in S\} \cup \{\omega_{J,S}^T(\textbf{v}_{-T}) \} \]

If $v_i(J)$ and $v_j(J)$ were both less (or equal) to $\varphi(\textbf{v})$ or both more (or equal) to $\varphi(\textbf{v})$ we obtain that $S=S'$. If we have $v_i(J) < \varphi(\textbf{v}) < v_j(J)$ then we get $S' = (S \cup \{i\}) -\{j\}$. We follows have that if $S$ and $S'$ have the same cardinal-partition then $\omega_{J,S}^T=\omega_{J,S'}^T$.

Let us consider $v_i(J),v_j(J) \in \{\circ,\otimes\}$. We have $\mathcal{D''}_J(\textbf{v}) = \mathcal{D''}_J(\textbf{w})$, therefore by using the lemma \ref{useful lemma 2} we have that $\omega_{J,S}^{\mathcal{D''}_J}(\textbf{v}_{-\mathcal{D''}_J})=\omega_{J,S}^{\mathcal{D''}_J}(\textbf{w}_{-\mathcal{D''}_J})$. The same can be said when $v_i(J)$ and $w_i(J)$ are both in $\emptyset$. As such the $\omega_{J,S}^T$ functions are therefore Anonymous over the set of voters not in $T$. Without loss of generality, we can extend this to all $\omega_{J,S}^T$ functions are anonymous. 

It remains to be able to swap when $v_i(J)\in\{\circ,\otimes\}$ and $v_j(J) \in \mathcal{A}$. We have that for $j \not \in S$, $\omega_{J,S}^{\mathcal{D''}_J(\textbf{v})}(\textbf{v}_{-\mathcal{D''}_J}) = \omega_{J,S}^{\mathcal{D''}_J(\textbf{w})}(\textbf{w}_{-\mathcal{D''}_J})=\omega_{J,S}^{\mathcal{D''}_J(\textbf{v})}(\textbf{w}_{-\mathcal{D''}_J})$. It follows that $\omega_{J,S}^{\mathcal{D''}_J(\textbf{v})}=\omega_{J,S}^{\mathcal{D''}_J(\textbf{w})}$.
If $j\in S$ then $\omega_{J,S}^{\mathcal{D''}_J(\textbf{v})}(\textbf{v}_{-\mathcal{D''}_J}) = \omega_{J,(S\cup\{j\})- \{i\}}^{\mathcal{D''}_J(\textbf{w})}(\textbf{w}_{-\mathcal{D''}_J})=\omega_{J,(S\cup\{j\})- \{i\}}^{\mathcal{D''}_J(\textbf{v})}(\textbf{w}_{-\mathcal{D''}_J})$. Therefore if $T$ and $T$, $S$ and $S'$ have the same cardinal-partition then $\omega_{J,S}^T=\omega_{J,S'}^{T'}$.

We have obtained all the relevant information that provides our characterization.

$\Leftarrow:$ Suppose that we verify the characterization:
For any $\textbf{v}$, suppose we have $i$ and $j$ such that $\mathcal{C}_i=\mathcal{C}_j$. Let $\textbf{w}$ be obtained from $\textbf{v}$ by switching the ballots of $i$ and $j$ everything else equal. For any $J \in\mathcal{M}$:
\begin{itemize}
    \item If $v_i(J),v_j(J) \in \mathcal{A}$ then we have $\mathcal{D''}_J(\textbf{v})=\mathcal{D''}_J(\textbf{w})=T$, $\textbf{v}_{-\mathcal{D''}_J}(J) = \textbf{w}_{-\mathcal{D''}_J}(J)$. And all $S$, $S'$ that have the same cardinal-partition verify $\omega_{J,S}^T=\omega_{J,S}^T$. Therefore
    
\begin{align*}
    \varphi(\textbf{w})(J) = & \max_{S \subseteq T} \min \{w_k: k \in S\} \cup \{\omega_{J,S}^T(\textbf{w}_{-T}) \} \\
    = & \max_{S \subseteq T} \min \{v_k: k \in S\} \cup \{\omega_{J,S}^T(\textbf{v}_{-T}) \} \\
    = & \varphi(\textbf{v})(J).
\end{align*}

\item If $v_i(J),v_j(J)\in\{\circ,\otimes\}$ then we have $\mathcal{D''}_J(\textbf{v})=\mathcal{D''}_J(\textbf{w})=T$ and $\textbf{v}(J) =\textbf{w}(J)$ and the $\omega_{J,S}^T$ functions are anonymous. Therefore:

\begin{align*}
    \varphi(\textbf{w})(J) = & \max_{S \subseteq T} \min \{w_k: k \in S\} \cup \{\omega_{J,S}^T(\textbf{w}_{-T}) \} \\
    = & \max_{S \subseteq T} \min \{v_k: k \in S\} \cup \{\omega_{J,S}^T(\textbf{v}_{-T}) \} \\
    = & \varphi(\textbf{v})(J).
\end{align*}

\item If $v_i(J)\in\{\circ,\otimes\}$ and $v_j(J)\in \mathcal{A}$ then $\mathcal{D''}_J(\textbf{w})=T$ and $\mathcal{D''}_J(\textbf{v})=(T \cup \{j\}) -\{i\}=T'$ have the same partition-cardinal and $S$ and $(S \cup \{j\}) -\{i\}=S'$ have the same partition-cardinal. As such
 $\omega_{J,S}^{T}  = \omega_{J,S'}^{T'}.$
\begin{align*}
    \varphi(\textbf{w})(J) = & \max_{S' \subseteq T'} \min \{w_k: k \in S'\} \cup \{\omega_{J,S'}^T(\textbf{w}_{-T'}) \} \\
    = & \max_{S \subseteq T} \min \{v_k: k \in S\} \cup \{\omega_{J,S}^T(\textbf{v}_{-T}) \} \\
    = & \varphi(\textbf{v})(J).
\end{align*}
\end{itemize}

\end{proof}

\subsection{Proof for A proposition \ref{A prop}}\label{A prop proof}

\begin{proof}
$\Rightarrow:$ Suppose that the phantom-proxy mechanism $\psi$ is (A):
For any $\textbf{v}$ such that $\mathcal{C}_i=\mathcal{C}_j$. Let $\textbf{w}$ be obtained from $\textbf{v}$ by swapping $i$ and $j$
\begin{align*}
    \psi(\textbf{v})(J) = \mu_{g_J(\#\mathcal{D''}_J(\textbf{v}) + \#\mathcal{F}(\textbf{v}))}(\textbf{v}(J)\cup \mathcal{F}_J(\textbf{v}))
\end{align*}

We know that the $\omega_{J,S}^T$ functions are anonymous therefore we have that $f_{i,J}(v_i) = f_{j,J}(w_j)=f_{j,J}(v_i)$. It follows that $f_{i,J} = f_{j,J}$. This shows the existence of $f_{M,J}$. 

$\Leftarrow :$ Suppose that the phantom-proxy mechanism $\psi$ is such that for all $J$ and all $M$ and $i$ if $\mathcal{C}_i =M$ then $f_{i,J}=f_{i,J}$.
If we switch $v_i$ with $w_i$ then at worst we replace a proxy vote by the same proxy vote. Therefore we have that $\psi$ verifies anonymity.
\end{proof}

\subsection{Proof for SA theorem \ref{SA theo}}\label{SA theo proof}

\begin{proof}
$\Rightarrow:$ By the same reasoning as for theorem \ref{SP theo} we have the existence of $\omega_{J,d}^T$ functions such that $\forall \textbf{v}, \varphi(\textbf{v})(J) = med(\textbf{v}(J),\omega_{J,0}^{\mathcal{D''}_J}(\textbf{v}_{-\mathcal{D''}_J}),\dots,\omega_{J,n}^{\mathcal{D''}_J}(\textbf{v}_{-\mathcal{D''}_J}))$.

For any $\textbf{v}$ let $\textbf{w}$ be obtained from $\textbf{v}$ by switching the ballots of $i$ and $j$. 

Thanks to lemma \ref{useful lemma 2} for all $k$ and $J$, without loss of generality, we can assume  $\omega_{J,k}^{\mathcal{D''}_J(\textbf{v})}(\textbf{v}_{-{\mathcal{D''}_J(\textbf{v})}}) = \omega_{J,k}^{\mathcal{D''}_J(\textbf{w})}(\textbf{w}_{-{\mathcal{D''}_J(\textbf{w})}}) $.

By considering $v_i(J)$ and $w_i(J)$ in $\mathcal{E}$ we can deduce that the $\omega_{J,k}^{\mathcal{D''}_J}$ functions are strong Anonymous (without loss of generality, this includes the voters in $\mathcal{D''}_J$).

By considering $v_i(J) \in \mathcal{A}$ and $w_i(J) \in \mathcal{E}$:
\[\omega_{J,k}^{\mathcal{D''}_J(\textbf{v})}(\textbf{v}_{-{\mathcal{D''}_J(\textbf{v})}}) =
\omega_{J,k}^{\mathcal{D''}_J(\textbf{v})}(\textbf{w}_{-\mathcal{D''}_J(\textbf{v})}) =
\omega_{J,k}^{\mathcal{D''}_J(\textbf{w})}(\textbf{w}_{-\mathcal{D''}_J(\textbf{w})}) .\]

We can therefore conclude that $\omega_{J,k}^{\mathcal{D''}_J(\textbf{v})}=\omega_{J,k}^{\mathcal{D''}_J(\textbf{w})}$. Since any $i,j\in \mathcal{N}$ can be considered when switching from $\textbf{v}$ to $\textbf{w}$ we therefore have the existence of a set of functions $\omega_{J,k}^d$ such that if $\mathcal{D''}_J(\textbf{w}) = d$ then $\omega_{J,k}^{\mathcal{D''}_J(\textbf{v})} = \omega_{J,k}^d$.

$\Leftarrow:$ Let the $\omega_{J,k}^d$ functions be strongly anonymous and let $\varphi$ be defined as follows:
\[\forall \textbf{v}, \varphi(\textbf{v})(J) = med\{\textbf{v}(J),\omega_{J,0}^{\#\mathcal{D''}_J}(\textbf{v}_{-\mathcal{D''}_J}),\dots \omega_{J,\#\mathcal{D''}_J}^{\#\mathcal{D''}_J}(\textbf{v}_{-\mathcal{D''}_J}))\}\]

We therefore have that $\varphi$ is (SP,SA).
\end{proof}

\subsection{Proof for SA prop \ref{SA prop}} \label{SA prop proof}

\begin{proof}
$\Rightarrow:$ Suppose that $\psi$ is (SA) then if we switch $i$ and $j$ such that $v_i(J)$ and $v_j(J)$ are in $\mathcal{E}$ we obtain that $f_{i,J}(v_i)=f_{j,J}(v_i)$. Hence the existence of $f_J$.

$\Leftarrow:$ Suppose that the function $f_J$ exists, if we switch $i$ and $j$, at worst we are replacing a proxy vote by an identical proxy vote. Therefore we are strongly anonymous.
\end{proof}

\subsection{Proof for OC theorem \ref{OC theo}}\label{OC theo proof}

\begin{proof}
$\Rightarrow:$ 

\begin{enumerate}
\item If we have $\omega_{J,S}^T(\textbf{v}_{-T})= \omega_{J,S}^{T'}(\textbf{w}_{-T'})$ then for an appropriate choice of $\textbf{t}(J)$ values we get $\alpha = \varphi(\textbf{v})(J)=\omega_{J,S}^T(\textbf{v}_{-T})= \omega_{J,S}^{T'}(\textbf{w}_{-T'})=\varphi(\textbf{w})(J)$ that does not belong to $\textbf{t}(J)$ unless $\alpha \in \{\inf\mathcal{A},\sup\mathcal{A}\}$. Therefore by OC we get $\varphi(\textbf{t})(J)=\alpha$, it follows that $\omega_{J,S\cup S'}^{T\cup T'}(\textbf{t}_{-(T\cup T')}) =\alpha$.
\item Suppose that we have $\max(\omega_{J,\emptyset}^T(\textbf{v}_{-T}), \omega_{J,\emptyset}^{T'}(\textbf{w}_{-T'})) \leq \alpha \leq \min(\omega_{J,S}^T(\textbf{v}_{-T}), \omega_{J,S}^{T'}(\textbf{w}_{-T'}))$. Any changes to the grades for $J$ does not affect the phantom-mappings. For all $i\in S \cup S'$ we take $t_i(J) = \alpha$ and for all others we take $t_i(J) = \beta \leq \alpha$. We have $\varphi(\textbf{v})(J)=\varphi(\textbf{w})(J)=\alpha$, therefore by OC we have that $\varphi(\textbf{t})(J)=\alpha$. Suppose we could take $\beta < \alpha$.
For any $U \not \subseteq S \cup S'$ we have that $\min \{t_i(J) : i \in U\} = \beta$, as such:

\begin{align*}
    \alpha=\varphi(\textbf{t})(J)= & \max_{U \subseteq T \cup T'}\min \{t_i(J) : i \in U\} \cup \{\omega_{J,U}^{T \cup T'}(\textbf{t}_{-(T\cup T')}) \} \\
    = & \max\{\min \{\alpha,\omega_{S \cup S'}^{T \cup T'}(\textbf{t}_{-(T\cup T')}) \},\omega_{J,\emptyset}^{T \cup T'}(\textbf{t}_{-(T\cup T')})\} \\
\end{align*}

Since $\omega_{J,\emptyset}^{T \cup T'}(\textbf{t}_{-(T\cup T')})$ was not selected we therefore have that $\omega_{J,\emptyset}^{T \cup T'}(\textbf{t}_{-(T\cup T')})\leq \alpha$. As such:

\[\varphi(\textbf{v})(J) =\min \{\alpha,\omega_{S \cup S'}^{T \cup T'}(\textbf{t}_{-(T\cup T')}) \}\]

Therefore we have $\alpha \leq \omega_{S \cup S'}^{T \cup T'}(\textbf{t}_{-(T\cup T')})$. 

If it was impossible to take $\beta < \alpha$, then we have $\alpha=\inf\mathcal{A}$. Since we take the $\max$ is the characterization we have $\omega_{J,\emptyset}^{T\cup T'}(\textbf{t}_{-(T\cup T')}) \leq \alpha$. By the lemma \ref{useful lemma 1} we have $\alpha \leq \omega_{S \cup S'}^{T \cup T'}(\textbf{t}_{-(T\cup T')})$.
\item Suppose that we have $\max(\omega_{J,S}^T(\textbf{v}_{-T}), \omega_{J,S'}^{T'}(\textbf{w}_{-T'})) \leq \alpha \leq \min(\omega_{J,T}^T(\textbf{v}_{-T}), \omega_{J,T}^{T'}(\textbf{w}_{-T'})) $. Any changes to the grades for $J$ does not affect the phantom-mappings. For all $i\not \in S \cup S'$ we take $t_i(J) = \alpha$ and for all $i\in S \cup S'$ we take $t_i(J) = \beta \geq \alpha$. We have $\varphi(\textbf{v})(J)=\varphi(\textbf{w})(J) = \alpha$. Therefore by IC we get that $\varphi(\textbf{t})(J) =\alpha$. Suppose that we can find $\beta > \alpha$.
If we had $\omega_{S \cup S'}^{T \cup T'}(\textbf{t}_{-(T \cup T')}) > \alpha$ then we would have $\min \{t_i(J) : i \in S \cup S'\} \cup \{\omega_{J,S \cup S'}^{T \cup T'}(\textbf{t}_{-(T \cup T')}) \} > \varphi(\textbf{t})(J)$. This is absurd, as such $\omega_{J,S \cup S'}^{T \cup T'}(\textbf{t}_{-(T \cup T')}) \leq \alpha$. If $\omega_{J,T \cup T'}^{T \cup T'}(\textbf{t}_{-(T \cup T')}) < \alpha$ then for all 
$U$ we have $\omega_{J,U}^{T \cup T'}(\textbf{t}_{-(T \cup T')}) < \alpha$ therefore we would have $\varphi(\textbf{v})(J) < \alpha$. This is absurd, as such $\alpha \leq \omega_{J,S \cup S'}^{T \cup T'}(\textbf{t}_{-(T \cup T')})$.

If we could not find $\beta > \alpha$ then we have that $\alpha=\sup\mathcal{A}$. Due to the lemma \ref{useful lemma 1} If we have $\omega_{J,S \cup S'}^{T \cup T'}(\textbf{t}_{-(T \cup T')}) > \alpha$ then we have $\omega_{J,\emptyset}^{T \cup T'}(\textbf{t}_{-(T \cup T')}) > \alpha$ we therefore contradict $\varphi(\textbf{t})(J)=\alpha$. Similarly $\omega_{J,T\cup T'}^{T \cup T'}(\textbf{t}_{-(T \cup T')}) < \alpha$ implies that $\varphi(\textbf{t})(J)<\alpha$. As such $\omega_{J,S \cup S'}^{T \cup T'}(\textbf{t}_{-(T \cup T')}) \leq \alpha \leq \omega_{J,T\cup T'}^{T \cup T'}(\textbf{t}_{-(T \cup T')})$.
\end{enumerate}

$\Leftarrow:$ Let us suppose that we have $\varphi$ that verifies the characterization. Suppose that $\varphi(\textbf{v})(J)=\varphi(\textbf{w})(J)$:

\begin{align*}
    \varphi(\textbf{t})(J) = & 
    \max_{S_v \subseteq \mathcal{D''}_J(\textbf{v}), S_w \subseteq \mathcal{D''}_J(\textbf{w})} \min \{t_i(J) : i \in S_v \cup S_w\} \cup \{\omega_{J,S_v \cup S_w}^{\mathcal{D''}_J(\textbf{t})}(\textbf{t}_{-\mathcal{D''}_J(\textbf{z)}}) \} \\
    = & \max_{S_v \subseteq \mathcal{D''}_J(\textbf{v}), S_w \subseteq \mathcal{D''}_J(\textbf{w})} \min \{v_i(J) : i \in S_v\} \cup \{w_i(J) : i \cup S_w\} \cup \{\omega_{J,S_v \cup S_w}^{\mathcal{D''}_J(\textbf{t})}(\textbf{t}_{-\mathcal{D''}_J(\textbf{z)}}) \} \\
    \geq & \min \{v_i(J) : i \in S\} \cup \{w_i(J) : i \cup S'\} \cup \{\omega_{J,S \cup S}^{\mathcal{D''}_J(\textbf{t})}(\textbf{t}_{-\mathcal{D''}_J(\textbf{z)}}) \} \\
    \geq & \min \{v_i(J) : i \in S\} \cup \{w_i(J) : i \cup S'\} \cup \{\min(\omega_{J,S}^{\mathcal{D''}_J(\textbf{v})}(\textbf{v}_{-\mathcal{D''}_J(\textbf{v})}),\omega_{J,S'}^{\mathcal{D''}_J(\textbf{w})}(\textbf{w}_{-\mathcal{D''}_J(\textbf{w)}}))\} \\
    \geq & \varphi(\textbf{v}))(J)
\end{align*}

\begin{align*}
    \varphi(\textbf{t})(J) = & 
    \max_{S_v \subseteq \mathcal{D''}_J(\textbf{v}), S_w \subseteq \mathcal{D''}_J(\textbf{w})} \min \{t_i(J) : i \in S_v \cup S_w\} \cup \{\omega_{J,S_v \cup S_w}^{\mathcal{D''}_J(\textbf{t})}(\textbf{t}_{-\mathcal{D''}_J(\textbf{z)}}) \} \\
    \leq & \max_{S_v \subseteq \mathcal{D''}_J(\textbf{v}), S_w \subseteq \mathcal{D''}_J(\textbf{w})} \min \{t_i(J) : i \in S_v \cup S_w\} \cup \{\max(\omega_{J,S_v}^{\mathcal{D''}_J(\textbf{v})}(\textbf{v}_{-\mathcal{D''}_J(\textbf{v})}),\omega_{J,S_w}^{\mathcal{D''}_J(\textbf{w})}(\textbf{w}_{-\mathcal{D''}_J(\textbf{w)}}))\} \\
    \leq & \max_{S_v \subseteq \mathcal{D''}_J(\textbf{v})} \min \{v_i(J) : i \in S_v\} \cup \{\max(\omega_{J,S_v}^{\mathcal{D''}_J(\textbf{v})}(\textbf{v}_{-\mathcal{D''}_J(\textbf{v})}),\omega_{J,S_w}^{\mathcal{D''}_J(\textbf{w})}(\textbf{w}_{-\mathcal{D''}_J(\textbf{w)}}))\} \\
    \varphi(\textbf{t})(J) \leq & \max_{S_w \subseteq \mathcal{D''}_J(\textbf{w})} \min \{w_i(J) : i \in S_w\} \cup \{\max(\omega_{J,S_v}^{\mathcal{D''}_J(\textbf{v})}(\textbf{v}_{-\mathcal{D''}_J(\textbf{v})}),\omega_{J,S_w}^{\mathcal{D''}_J(\textbf{w})}(\textbf{w}_{-\mathcal{D''}_J(\textbf{w)}}))\}
\end{align*}

Therefore $\varphi(\textbf{t})(J)=\varphi(\textbf{v})(J)$.
\end{proof}

\subsection{Proof for OC proposition \ref{OC prop}}\label{OC prop proof}

\begin{proof}
$\Rightarrow:$
Let us consider a partition $\mathcal{N}_1 \cup \mathcal{N}_2 =\mathcal{N}$ of the set of voters. For any voting profile $\textbf{t}$, let $\textbf{v}=\textbf{t}_{-(\mathcal{N}_2)}$ and $\textbf{w}=\textbf{t}_{-(\mathcal{N}_1)}$. We define $k_1=\#(\textbf{v}(J) \cup \mathcal{F}(\textbf{v}))$ and $k_2=\#(\textbf{w}(J) \cup \mathcal{F}(\textbf{w}))$.

We know that the $g_J(k_1)$-th and $g_J(k_2)$-th smallest members the voting pools $(\textbf{v}(J) \cup \mathcal{F}(\textbf{v}))$ and $(\textbf{w}(J) \cup \mathcal{F}(\textbf{w}))$ respectively have the same value. $g_J$ has no way on ensuring that any other element of these voting pools have the same value. As such when we merge the voting pools we must select one of those 2 elements. It follows that $g_J(k_1+k_2) \in \{g_J(k_1) + g_J(k_2)-1,g_J(k_1) + g_J(k_2)\}$.

$\Leftarrow:$

Suppose that for all $p_1, p_2$ we have $g_J(p_1 + p_2) \in \{g_J(p_1) +g_J(p_2)-1,g_J(p_1) +g_J(p_2)\}$. Then it is immediate that when we merge two voter pools that provided the same outcome, we obtain said outcome.

 \end{proof}

\subsection{Proof for IC theorem \ref{IC theo}}\label{IC theo proof}

\begin{proof}
$\Rightarrow:$ 

\begin{enumerate}
\item If we have $\omega_{J,S}^T(\textbf{v}_{-T})= \omega_{J,S}^{T'}(\textbf{w}_{-T'})$ then for an appropriate choice of $\textbf{z}(J)$ values we get $\alpha = \varphi(\textbf{v})(J)=\omega_{J,S}^T(\textbf{v}_{-T})= \omega_{J,S}^{T'}(\textbf{w}_{-T'})=\varphi(\textbf{w})(J)$ that does not belong to $\textbf{z}(J)$ unless $\alpha \in \{\inf\mathcal{A},\sup\mathcal{A}\}$. Therefore by IC we get $\varphi(\textbf{z})(J)=\alpha$, it follows that $\omega_{J,S\cup S'}^{T\cup T'}(\textbf{z}_{-(T\cup T')}) =\alpha$.
    \item Suppose that we have $\max(\omega_{J,\emptyset}^T(\textbf{v}_{-T}), \omega_{J,\emptyset}^{T'}(\textbf{w}_{-T'})) \leq \alpha \leq \min(\omega_{J,S}^T(\textbf{v}_{-T}), \omega_{J,S}^{T'}(\textbf{w}_{-T'}))$. Any changes to the grades for $J$ does not affect the phantom-mappings. For all $i\in S \cup S'$ we take $t_i(J) = \alpha$ and for all others we take $t_i(J) = \beta \leq \alpha$. We have $\varphi(\textbf{v})(J)=\varphi(\textbf{w})(J)=\alpha$, therefore by (IC) we have that $\varphi(\textbf{z})(J)=\alpha$ Suppose we could take $\beta < \alpha$.
For any $U \not \subseteq S \cup S'$ we have that $\min \{z_i(J) : i \in U\} = \beta$, as such:

\begin{align*}
    \alpha=\varphi(\textbf{z})(J)= & \max_{U \subseteq T \cup T'}\min \{z_i(J) : i \in U\} \cup \{\omega_{J,U}^{T \cup T'}(\textbf{z}_{-(T\cup T')}) \} \\
    = & \max\{\min \{\alpha,\omega_{S \cup S'}^{T \cup T'}(\textbf{z}_{-(T\cup T')}) \},\omega_{J,\emptyset}^{T \cup T'}(\textbf{z}_{-(T\cup T')})\} \\
\end{align*}

Since $\omega_{J,\emptyset}^{T \cup T'}$ was not selected we therefore have that $\omega_{J,\emptyset}^{T \cup T'}\leq \alpha$. As such:

\[\varphi(\textbf{v})(J) =\min \{\alpha,\omega_{S \cup S'}^{T \cup T'}(\textbf{z}_{-(T\cup T')}) \}\]

Therefore we have $\alpha \leq \omega_{S \cup S'}^{T \cup T'}(\textbf{z}_{-(T\cup T')})$. 

If it was impossible to take $\beta < \alpha$, then we have $\alpha=\inf\mathcal{A}$. Since we take the $\max$ is the characterization we have $\omega_{J,\emptyset}^{T\cup T'}(\textbf{z}_{-(T\cup T')}) \leq \alpha$. Then by the lemma \ref{useful lemma 1}, we have $\alpha \leq \omega_{S \cup S'}^{T \cup T'}(\textbf{z}_{-(T\cup T')})$.
\item Suppose that we have $\max(\omega_{J,S}^T(\textbf{v}_{-T}), \omega_{J,S'}^{T'}(\textbf{w}_{-T'})) \leq \alpha \leq \min(\omega_{J,T}^T(\textbf{v}_{-T}), \omega_{J,T}^{T'}(\textbf{w}_{-T'})) $. Any changes to the grades for $J$ does not affect the phantom-mappings. For all $i\not \in S \cup S'$ we take $z_i(J) = \alpha$ and for all $i\in S \cup S'$ we take $z_i(J) = \beta \geq \alpha$. We have $\varphi(\textbf{v})(J)=\varphi(\textbf{w})(J) = \alpha$. Therefore by IC we get that $\varphi(\textbf{z})(J) =\alpha$. Suppose that we can find $\beta > \alpha$.
If we had $\omega_{S \cup S'}^{T \cup T'}(\textbf{z}_{-(T \cup T')}) > \alpha$ then we would have $\min \{z_i(J) : i \in S \cup S'\} \cup \{\omega_{J,S \cup S'}^{T \cup T'}(\textbf{z}_{-(T \cup T')}) \} > \varphi(\textbf{z})(J)$. This is absurd, as such $\omega_{J,S \cup S'}^{T \cup T'}(\textbf{z}_{-(T \cup T')}) \leq \alpha$. If $\omega_{J,T \cup T'}^{T \cup T'}(\textbf{z}_{-(T \cup T')}) < \alpha$ then for all 
$U$ we have $\omega_{J,U}^{T \cup T'}(\textbf{z}_{-(T \cup T')}) < \alpha$ therefore we would have $\varphi(\textbf{v})(J) < \alpha$. This is absurd, as such $\alpha \leq \omega_{J,S \cup S'}^{T \cup T'}(\textbf{z}_{-(T \cup T')})$.

If we could not find $\beta > \alpha$ then we have that $\alpha=\sup\mathcal{A}$. According to the lemma \ref{useful lemma 1}, if $\omega_{J,S \cup S'}^{T \cup T'}(\textbf{z}_{-(T \cup T')}) > \alpha$ then $\omega_{J,\emptyset}^{T \cup T'}(\textbf{z}_{-(T \cup T')}) > \alpha$. We therefore contradict $\varphi(\textbf{z})(J)=\alpha$. Similarly $\omega_{J,T\cup T'}^{T \cup T'}(\textbf{z}_{-(T \cup T')}) < \alpha$ implies that $\varphi(\textbf{z})(J)<\alpha$. As such $\omega_{J,S \cup S'}^{T \cup T'}(\textbf{z}_{-(T \cup T')}) \leq \alpha \leq \omega_{J,T\cup T'}^{T \cup T'}(\textbf{z}_{-(T \cup T')})$.
\end{enumerate}

$\Leftarrow:$ Let us suppose that we have $\varphi$ that verifies the characterization. Suppose that $\varphi(\textbf{v})(J)=\varphi(\textbf{w})(J)$:

\begin{align*}
    \varphi(\textbf{z})(J) = & 
    \max_{S_v \subseteq \mathcal{D''}_J(\textbf{v}), S_w \subseteq \mathcal{D''}_J(\textbf{w})} \min \{z_i(J) : i \in S_v \cup S_w\} \cup \{\omega_{J,S_v \cup S_w}^{\mathcal{D''}_J(\textbf{z})}(\textbf{z}_{-\mathcal{D''}_J(\textbf{z)}}) \} \\
    = & \max_{S_v \subseteq \mathcal{D''}_J(\textbf{v}), S_w \subseteq \mathcal{D''}_J(\textbf{w})} \min \{v_i(J) : i \in S_v\} \cup \{w_i(J) : i \cup S_w\} \cup \{\omega_{J,S_v \cup S_w}^{\mathcal{D''}_J(\textbf{z})}(\textbf{z}_{-\mathcal{D''}_J(\textbf{z)}}) \} \\
    \geq & \min \{v_i(J) : i \in S\} \cup \{w_i(J) : i \cup S'\} \cup \{\omega_{J,S \cup S}^{\mathcal{D''}_J(\textbf{z})}(\textbf{z}_{-\mathcal{D''}_J(\textbf{z)}}) \} \\
    \geq & \min \{v_i(J) : i \in S\} \cup \{w_i(J) : i \cup S'\} \cup \{\min(\omega_{J,S}^{\mathcal{D''}_J(\textbf{v})}(\textbf{v}_{-\mathcal{D''}_J(\textbf{v})}),\omega_{J,S'}^{\mathcal{D''}_J(\textbf{w})}(\textbf{w}_{-\mathcal{D''}_J(\textbf{w)}}))\} \\
    \geq & \varphi(\textbf{v}))(J)
\end{align*}

\begin{align*}
    \varphi(\textbf{z})(J) = & 
    \max_{S_v \subseteq \mathcal{D''}_J(\textbf{v}), S_w \subseteq \mathcal{D''}_J(\textbf{w})} \min \{z_i(J) : i \in S_v \cup S_w\} \cup \{\omega_{J,S_v \cup S_w}^{\mathcal{D''}_J(\textbf{z})}(\textbf{z}_{-\mathcal{D''}_J(\textbf{z)}}) \} \\
    \leq & \max_{S_v \subseteq \mathcal{D''}_J(\textbf{v}), S_w \subseteq \mathcal{D''}_J(\textbf{w})} \min \{z_i(J) : i \in S_v \cup S_w\} \cup \{\max(\omega_{J,S_v}^{\mathcal{D''}_J(\textbf{v})}(\textbf{v}_{-\mathcal{D''}_J(\textbf{v})}),\omega_{J,S_w}^{\mathcal{D''}_J(\textbf{w})}(\textbf{w}_{-\mathcal{D''}_J(\textbf{w)}}))\} \\
    \leq & \max_{S_v \subseteq \mathcal{D''}_J(\textbf{v})} \min \{v_i(J) : i \in S_v\} \cup \{\max(\omega_{J,S_v}^{\mathcal{D''}_J(\textbf{v})}(\textbf{v}_{-\mathcal{D''}_J(\textbf{v})}),\omega_{J,S_w}^{\mathcal{D''}_J(\textbf{w})}(\textbf{w}_{-\mathcal{D''}_J(\textbf{w)}}))\} \\
    \varphi(\textbf{z})(J) \leq & \max_{S_w \subseteq \mathcal{D''}_J(\textbf{w})} \min \{w_i(J) : i \in S_w\} \cup \{\max(\omega_{J,S_v}^{\mathcal{D''}_J(\textbf{v})}(\textbf{v}_{-\mathcal{D''}_J(\textbf{v})}),\omega_{J,S_w}^{\mathcal{D''}_J(\textbf{w})}(\textbf{w}_{-\mathcal{D''}_J(\textbf{w)}}))\}
\end{align*}

Therefore $\varphi(\textbf{z})(J)=\varphi(\textbf{v})(J)$.
\end{proof}

\section{Section Ranking}
\subsection{Proof for determinism property \ref{R_J deter}}

\begin{proof}\label{R_J deter proof}
Let $P_J^\psi:\mathcal{B}^n \rightarrow \mathcal{B}^n$ be function that turns a multi-set $\mathcal{V}$ in an ordered-set defined as:

 \begin{enumerate}
     \item $P_J^{\psi} := []$; $S = \mathcal{V}$;
     \item If $S = \emptyset$; Return $P_J^{\psi}$, END \label{P_G start loop}
     \item Else $\alpha :=\mu_{g(\#S)}(S)$ $P_J^{\psi} := (R_J^{\psi},\alpha)$; $S = S-\alpha$;
     \item Go back to (\ref{P_G start loop})
\end{enumerate}

By definition of a multi-set we have that $P_J^{\psi}$ is well defined.

Let us now show that for we have:
\[\forall \textbf{v}, R_J^{\psi}(\textbf{v})=P_J^{\psi}(\textbf{v}(J) \cup \mathcal{F}_J(\textbf{v})).\]

Recall that $\psi$ is defined as $\mu_{g(\# \textbf{v}(J) \cup \mathcal{F}_J(\textbf{v}))}(\textbf{v}(J) \cup \mathcal{F}_J(\textbf{v}))$.

It is therefore sufficient to show that for any $i$ we have that if $\textbf{w}=\textbf{v}[\forall J, v_i(J) = \emptyset]$ then if $v_i(J) \in \mathcal{A}$ then $\textbf{w}(J) \cup \mathcal{F}_J(\textbf{w}) = (\textbf{v}(J) \cup \mathcal{F}_J(\textbf{v})) - \{v_i(J)\}$ and if $v_i(J) \not \in \mathcal{A}$ we have $\textbf{w}(J) \cup \mathcal{F}_J(\textbf{w}) = (\textbf{v}(J) \cup \mathcal{F}_J(\textbf{v})) - \{f_{i,J}(J)\}$.

\begin{itemize}
    \item If $v_i(J) \in \mathcal{A}$ then since $f_{i,J}(v_i)[\forall J, v_i(J) = \emptyset] =\emptyset)$ (by definition of the proxy functions) and since all the other proxy functions do not depend on $v_i$ we have $\mathcal{F}_J(\textbf{v}) = \mathcal{F}_J(\textbf{w})$. $\textbf{v}(J) = \textbf{w}(J) \cup \{v_i(J)\}$. Therefore we have the equality.
    \item If $f_{i,J}(v_i(J)) =\emptyset$ then $f_{i,J}(v_i(J)) = f_{i,J}(v_i(J)[\forall J, v_i(J) = \emptyset])$. Since none of the other proxy functions depend on $v_i$ we have $\mathcal{F}_J(\textbf{v}) = \mathcal{F}_J(\textbf{w}) - \{f_{i,J}(v_i)\}$. We also have $\textbf{v}(J) = \textbf{w}(J)$.
    \item If $f_{i,J}(v_i(J)) \in \mathcal{B}$ then since none of the other proxy function depend on $v_i$ we have $\mathcal{F}_J(\textbf{v}) = \mathcal{F}_J(\textbf{w}) - \{f_{i,J}(v_i)\}$. We also have $\textbf{v}(J) = \textbf{w}(J)$.
\end{itemize}

We therefore have the desired equality. This ends our proof.
\end{proof}

\subsection{Proof for total ranking \ref{ranking corol}}\label{ranking corol proof}

\begin{proof}
    Because $\psi$ is fair (F) we have that any 2 candidates that provide the same voting pool obtain the same ranking of their elements. If 2 voting pools then an ordering of their elements must also be different. The fact that the lexicographic order is a total order provides the rest. 
\end{proof}

\subsection{Proof for SP in the new space \ref{newSP theo}}\label{newSP theo proof}

\begin{proof}
    Suppose that we have $(v_i(J)) <_\psi R_J^\psi(\textbf{v})$ (with $v_i(J) \in \mathcal{A}$). Let $w_i$ be such $w_i(J) \in \mathcal{A}$.
    Let $\textbf{w}=\textbf{v}[v_i:=w_i]$.

    We have $\textbf{w}(J) = (\textbf{v}(J) \cup \{w_i(J)\}) - \{v_i(J))$ and $\mathcal{F}_J(\textbf{v}) = \mathcal{F}_J(\textbf{w})$.

    Let $k$ be the first value such that $R_J^(\textbf{v})(k) \neq R_J(\textbf{w})(J)$. 

    This is therefore the first time the loop could not choose the same element in the multi-set. As such since for all other voters the vote that represents them in the two voting pools is the same this is also the first time we cannot remove the same voter in both instances.

    Since $R_J^\psi$ is defined iterative we can therefore assume $k=1$ without loss of generality. It follows that we are comparing $\psi(\textbf{v})$ and $\psi(\textbf{w})$. Since $\psi$ is SP $\psi(\textbf{v}) \leq  \psi(\textbf{w})$ therefore we have $R_J^\psi(\textbf{v}) < R_J^\psi(\textbf{v})$. 
    
    This concludes the proof. (The case $v_i(J) > R_J^\psi(\textbf{v})$ is symmetrical).

\end{proof}

\subsection{Proof for the duplicate proposition \ref{duplicate prop}}\label{duplicate proof}

\begin{proof}
    Let $\psi$ be a phantom-proxy mechanism that verifies (OC). For any $\textbf{v}$ and $\textbf{w}$ such that $\textbf{w}$ is obtained by duplicating $\textbf{v}$ $k$ times ($k \geq 1$). Suppose that our proposition is false. That is to say that there are two voters $i$ and $j$ such that for $\textbf{v}$ the vote (or proxy-vote) of $i$ was selected by $R$ before the vote of $j$ and for $\textbf{w}$ the first time we selected a proxy for $j$ was before the first time we selected a proxy for $i$. Wlog we will assume that $j$ is the first voter whose duplicate got selected early and $i$ the first voter whose first duplicate got selected late. 
    
    Wlog we will assume that all votes are different. (If not we are careful in our order of selection of duplicate when we have the choice). 
    
    Let $\tilde{v_i(J)}$ (resp $\tilde{v_j(J)})$) be the vote that represents $i$ (resp $j$) in $\textbf{v}$. Let $n$ be the number of elements left in the voting pool when $g$ selects $\tilde{v_i(J)}$. Let $g(n) = r$.
    \begin{itemize}
        \item If $\tilde{v_i(J)}$ was selected as an element that is lower than $\tilde{v_j(J)}$. Let $T$ be the remaining voting pool when $\tilde{v_i(J)}$ was selected by $\textbf{v}$. Let $S \subset T$ be the remaining votes in the voting pool for $\textbf{v}$ when $\tilde{v_i(J)}$ was selected that were considered lower than $\tilde{v_i(J)}$. Let us now consider the moment a duplicate of $v_j(J)$ was selected. Any duplicates of an element of $S$ was never selected. There are $rk$ such duplicates. All of these elements are considered smaller than $\tilde{v_j(J)}$. There are also $k'\geq 0$ elements that are not a duplicate of an element of $T$ remaining in the voting pool of $\textbf{w}$. These elements are also considered smaller than $\tilde{v_j(J)}$. We therefore have $g(nk + k') = rk + k' +1$.
        By (OC) we have $g(nk + k') \leq kg(n) + k' =kr +k'$. We have reached a contradiction.
        \item If $\tilde{v_i(J)}$ was selected as an element that is greater than the $\tilde{v_j(J)}$. We have that $g(n) = g(n-1) + 1 =r$. Let $T$ be the remaining voting pool when $\tilde{v_i(J)}$ was selected by $\textbf{v}$. Let $S \subset T$ be the remaining votes in the voting pool for $\textbf{v}$ when $\tilde{v_i(J)}$ was selected that were considered greater than $\tilde{v_i(J)}$. Let us now consider the moment a duplicate of $v_j(J)$ was selected. Any duplicates of an element of $S$ was never selected. There are $(n-r+1)k$ such duplicates. All of these elements are considered greater than $\tilde{v_j(J)}$. There are also $k'\geq 0$ elements that are not a duplicate of an element of $T$ remaining in the voting pool of $\textbf{w}$. These elements are also considered greater than $\tilde{v_j(J)}$. We therefore have $g(nk + k') = k(r-1)$.
        By (OC) we have $g(nk + k') \geq g(kn) + g(k') \geq kg(n) -(k-1) +g(k') \geq rk -(k-1) = r(k-1) +1 $. We have reached a contradiction.
    \end{itemize}
    
    This concludes our proof.
\end{proof}






\subsection{Proof for the corollary \ref{duplicate corol}} \label{duplicate corol proof}

\begin{proof}
    Since $\psi$ is fair we have the same $g$ for $J$ and $I$. As such is we order our votes so the $k$-th smallest vote for $I$ correspond to the same voter as the $k$-th smallest vote for $J$. For all voters $i$, we have that each time a duplicate of $i$ is chosen in the voting pool for $I$, a duplicate for $i$ is selected for $J$. It follows that the order in which we discover a duplicate belonging to a new voter for the first time determines the ranking of $I$ compared to $J$. The property \ref{duplicate prop} for (OC) phantom-proxy mechanisms therefore concludes the proof.
\end{proof}






\section{Pareto optimality}\label{Pareto}

\begin{definition}[Pareto Optimal :OP]
    A voting method $\varphi: \mathcal{O}^{\mathcal{M} \times \mathcal{N}} \rightarrow {(\mathcal{B} \cup \{\emptyset\})}^\mathcal{M}$ is Pareto Efficient if there does not exist an outcome $\alpha \in \mathcal{B}$ such that for at least one voter $i$ that gave a grade $v_i(J)$ we have $\alpha$ is closer to $v_i(J)$ than $\varphi(\textbf{v})(J)$ and for no voters $j\in \mathcal{D''}_J$, $\alpha$ is further away from $v_j(J)$ than $\varphi(\textbf{v})(J)$.
\end{definition}

\begin{theorem}[Pareto Optimal characterization]
    A voting method $\varphi: \mathcal{O}^{\mathcal{M} \times \mathcal{N}} \rightarrow {(\mathcal{B} \cup \{\emptyset\})}^\mathcal{M}$ is Pareto Optimal in regards to candidates iff it is unanimous. 
\end{theorem}

\begin{proof}
    $\Rightarrow:$Suppose that we are Pareto Optimal in regards to candidates than if all voters gave the same grade $\alpha$ by Pareto Optimality the outcome must be $\alpha$. The method is therefore unanimous.

    $\Leftarrow:$Suppose that we are unanimous then let $i$ be the voter that gave $J$ the smallest grade $v_i(J) = \alpha$ and $j$ the voter that gave $J$ the largest grade $v_j(J)=\beta$. 
    Let $S$ be such that $\omega_{J,S}^{\mathcal{D''}_J}(\textbf{v}_{-\mathcal{D''}_J})$ is greater or equal to $v_i(J)$. 
    \[\min \{ v_k(J) : k \in S\} \cup \{\omega_{J,S}^{\mathcal{D''}_J}(\textbf{v}_{-\mathcal{D''}_J}) \} \geq v_i(J)\]
    Therefore we have that $\varphi(\textbf{v})(J) \geq v_i(J)$.
    When $ S \neq \emptyset$, we have that 
    \[\min \{ v_k(J) : k \in S\} \leq v_j(J)\]
    and $\omega_{J,\emptyset}^{\mathcal{D''}_J} \leq v_j(J)$ therefore $\varphi(\textbf{v})(J) \leq v_i(J)$.
    As such any alternative outcome will result in $i$ or $j$ getting further away from the outcome. We are Pareto optimal
\end{proof}

\section{Reinforcing absentees}
Let us now suggest a small modification that reinforces our concept of what absentee votes represent.

Ranking by using $\psi$ is considered fair because all voters are represented at most once in the voting pool and $\mu_{g(\#\mathcal{V}(J))}$ does not distinguish between elements of the voting pool. As such the mechanism is still considered fair if an element is added to the voting pool that represents a voter that was not yet represented in the voting pool.

The spirit of the SC rule is that if voter $i\in\mathcal{N}$ abstained then he might as well have voted for the outcome $\alpha$. When considering SI, phantom-proxy methods $\psi$ functions, if a voter abstained then that voter does not have a proxy vote. As such we remain fair in our model if we add an element worth $\alpha$ representing $i$ to the voting pool. This late addition to the voting pool represents voter $i$ consenting with the outcome. We call these votes the \textbf{absentee votes}. When $\mathcal{A}=\mathcal{B}$, we can represent this when ordering our candidates by using the S-grading range $P_J^\psi$ instead of the grading range to $R_J^\psi$ where $P_J^\psi$ is defined as:

 \begin{enumerate}
     \item $P_J^{\psi} := []$; $S =\emptyset$;
     \item $\textbf{w} := \textbf{v}[v_i(J):= \psi(\textbf{v}) : \forall i, v_i(J) =\circ]$;
     \item If $\psi(\textbf{w}[\forall J,\forall i \in S, w_i(J) = \emptyset])(J) = \emptyset$; Return $P_J^{\psi}$, END \label{T_G start loop}
     \item Else $\alpha :=\psi(\textbf{w}[\forall J,\forall i \in S, w_i(J) = \emptyset])(J)$ $P_J^{\psi} := (P_J^{\psi},\alpha)$;
     \item Find $i \in \mathcal{N}-S$ such that $w_i(J)=\alpha$ or $f_{i,J}(w_i) = \alpha$; $S: = S \cup \{i\}.$
     \item Go back to (\ref{T_G start loop})
\end{enumerate}

\begin{proposition}
    If $\forall i,v_i(J) \neq \circ$ then $P_J^\psi(\textbf{v}) = R_J^\psi(\textbf{v})$.
\end{proposition}

As such the S-grading range $P_J^\psi$ is an extension of the grading range $R_J^\psi$ that strengthens the meaning of an absentee vote.

\begin{proposition}
    If $R_J^\psi(\textbf{v}) < \psi(\textbf{v})(J)$ then we have $R_J^\psi(\textbf{v}) < P_J^\psi(\textbf{v}) < \psi(\textbf{v})(J)$ and if $R_J^\psi(\textbf{v}) > \psi(\textbf{v})(J)$ then we have $R_J^\psi(\textbf{v}) > P_J^\psi(\textbf{v}) > \psi(\textbf{v})(J)$ 
\end{proposition}

As such we have that the absentee voter is strengthening the outcome. The previous method has be summarized for a finite number of candidates in the algorithm \ref{algo ranking}.

\begin{algorithm}
\caption{The Ranking algorithm}\label{algo ranking}
\KwData{$\textbf{v}$ and $\psi$}
\KwResult{A ranking of the candidates}
\For{$J \in \mathcal{M}$}{$\mathcal{V}(J) \gets \textbf{v}(J) \cup \mathcal{F}_J(\textbf{v})$\;
\If{We wish to verify the reinforced absentee property}{\For{$i \in \{i : v_i(J) = \circ\}$}{$\mathcal{V}(J) \gets \mathcal{V}(J) \cup \psi(\textbf{v})(J)$\;}}
\If{$\mathcal{V}(J)=\emptyset$}{REMOVE $J$ from the ranking process.}
}
$k \gets \Pi_{J \in \mathcal{M}}\#\mathcal{V}(J)$\;
\For{$J \in \mathcal{M}$}{$\mathcal{V}(J) \gets$ $\dfrac{k}{\#\mathcal{V}(J)} \mathcal{V}(J)$\;
$R_J \gets []$\;
\While{$\mathcal{V}(J) \neq \emptyset$}{$R_J \gets (R_J,\mu_{g(\#\mathcal{V}(J))}(\mathcal{V}(J))$\;
$\mathcal{V}(J) \gets \mathcal{V}(J) - \mu_{g(\#\mathcal{V}(J))(\mathcal{V}(J))}$\;
}
}
Rank the candidates according to the $R_J$ ranking.

\end{algorithm}

\section{end}
\end{document}